\documentclass[a4paper]{lipics-lairez}
\pdfoutput=1
\nolinenumbers
\hideLIPIcs

\usepackage{mathtools}
\usepackage{csquotes}
\usepackage{booktabs}
\usepackage{needspace}

\allowdisplaybreaks[2]

\usepackage[
style=apa,
backend=biber,
alldates=year,
sortcites=false
]{biblatex}

\IfFileExists{maths.bib}{\addbibresource{maths.bib}}{}
\def\marginindex#1{}
\def\margindef#1{\emph{#1}\marginindex{#1}}
\def\margindef#1{\emph{#1}}

\DeclareFieldFormat{labelnumberwidth}{\color{lipicsGray}\sffamily\bfseries #1}

\usepackage[shortlabels]{enumitem}
\setlist[enumerate]{leftmargin=*,font=\upshape\bfseries\sffamily\color{lipicsGray}}
\setlist{topsep=\smallskipamount}

\usepackage{algorithm}
\usepackage[]{pseudo}

\counterwithout{pseudoline}{pseudoenv}
\makeatletter
\pseudoset{
  compact,
  label=\scriptsize\sffamily\color{lipicsGray}\arabic*,
  start=\thepseudoline+1,
  ctfont=\sffamily\color{lipicsGray},
  ct-left=\unskip\hspace{0pt plus 1filll}\hbox to 6cm\bgroup\ctfont{---},
  ct-right=\hss\egroup,
  begin-tabular =
  \tabularx{\linewidth}{@{}
    r 
    >{\pseudosetup} 
    X 
  },
  end-tabular = \endtabularx,
}

\usepackage{etoolbox}
\newcommand\defchunk[2]{%
  \protected@write\@auxout{}{\csgdef{chunk#1}{#2}}%
  \unskip« {\itshape\ifcsdef{chunk#1}{\csuse{chunk#1}}{}} » $\equiv$ \label{chunk#1}}
\newcommand\inschunk[1]{\unskip« {\itshape\ifcsdef{chunk#1}{\csuse{chunk#1}}{}} » $\to$ line \ref{chunk#1}}

\makeatother

\def\eqdef{\doteq}
\def\brack#1{\left\langle #1 \right\rangle}

\def\mop{\operatorname}
\def\st{\ \middle|\ }

\def\ref#1{{\upshape\sffamily\ref{#1}}}

\newcommand{\myto}[2][\to]{%
  \mathrel{\vbox{\offinterlineskip\ialign{%
        \hfil##\hfil\cr
        $\scriptscriptstyle#2$\cr
        \noalign{\kern0ex}
        $#1$\cr
      }}}}

\def\gb{Gröbner basis}
\def\gbs{Gröbner bases}
\def\rb{rewrite basis}
\def\rbs{rewrite bases}
\def\sb{signature basis}
\def\sbs{signature bases}

\def\lm{\mop{lm}}

\def\rew{\myto{1}}
\def\sig{\mop{sig}}
\def\tdown{\mathrel{\smash{\breve\downarrow}}}
\def\tup{\mathrel{\smash{\breve\uparrow}}}

\def\tailred{\myto[\smile]{1}}

\def\epsilon{\varepsilon}

\def\leq{\leqslant}
\def\geq{\geqslant}

\def\basis{\mathcal{M}}
\def\SIG{\mathcal{S}}

\def\rk{\mop{rk}}

\def\eqlt{\equiv_{\mathrm{lt}}}

\def\epsilon{\varepsilon}

\def\children{\mathit{children}}
\def\labell{L}

\newcommand\AXle[2][\sigma]{A#2^{\leq #1}}
\newcommand\AXlt[2][\sigma]{A#2^{< #1}}
\newcommand\LAXle[2][\sigma]{\brack{\smash{A#2^{\leq #1}}}}
\newcommand\LAXlt[2][\sigma]{\brack{\smash{A#2^{< #1}}}}

\newcommand\AGle[1][\sigma]{\AXle[#1]{G}}
\newcommand\AGlt[1][\sigma]{\AXlt[#1]{G}}
\newcommand\LAGle[1][\sigma]{\LAXle[#1]{G}}
\newcommand\LAGlt[1][\sigma]{\LAXlt[#1]{G}}

\newcommand{\implication}[2]{\proofsubparagraph{Proof that \ref{#1} implies \ref{#2}.}}

\title{Axioms for a theory of signature bases}
\Copyright{Pierre Lairez}
\authorrunning{P. Lairez}
\author{Pierre Lairez}{Université Paris-Saclay, Inria, 91120 Palaiseau, France}{pierre.lairez@inria.fr}{0000-0003-3756-0151}{}

\funding{ This work has been supported by the European Research Council under
  the European Union's Horizon Europe research and innovation programme, grant
  agreement 101040794 (10000~DIGITS); and by the ANR grant ANR-19-CE40-0018 (De
  Rerum Natura). }

\keywords{\gb{}, F5 algorithm, signature basis}

\ccsdesc[500]{Computing methodologies~Algebraic algorithms}

\begin{document}

\maketitle
\begin{abstract}
  Twenty years after the discovery of the F5 algorithm, Gröbner bases with signatures are still challenging to understand and to adapt to different settings. This contrasts with Buchberger's algorithm, which we can bend in many directions keeping correctness and termination obvious. I propose an axiomatic approach to Gröbner bases with signatures with the purpose of uncoupling the theory and the algorithms, giving general results applicable in many different settings (e.g. Gröbner for submodules, F4-style reduction, noncommutative rings, non-Noetherian settings, etc.),
  and extending the reach of signature algorithms.
\end{abstract}

%

\section{Introduction}

\subparagraph*{Context}
Introduced by \textcite{Faugere_2002} to compute \gbs{},
the F5 algorithm proposes the concept of \emph{signature}
to avoid the redundant computations that arise in Buchberger's algorithm \parencite{Buchberger_1965,Buchberger_2006}.
Each polynomial handled by the algorithm is augmented with a signature
designed to enforce a fundamental postulate, which we may state as ``two elements with the same signature are substitutable''.
We can find precursive ideas in the work of \textcite{GebauerMoller_1986} and \textcite{MollerMoraTraverso_1992}
and signatures also share somes ideas with Hilbert-driven algorithms \parencite[]{Traverso_1996}.

Today's situation of signature algorithms is equivocal.
F5's relevance, from the pure aspect of performance,
was demonstrated by a success on cryptographic challenges early on \parencite{FaugereJoux_2003}.
Moreover, the predictability of F5 in certain situations enables complexity analyses that are particularly relevant in cryptanalysis \parencite{BardetFaugereSalvy_2015}.
But none of the current best implementations for computing \gbs{} uses signatures, be it Magma \parencite{BosmaCannonPlayoust_1997}, msolve \parencite{BerthomieuEderSafeyElDin_2021}, Singular \parencite{Singular_CAS} or Maple.
They prefer Buchberger's algorithm, handling S-pairs as \textcite{GebauerMoller_1988} do
and using simultaneous reductions in the F4 style \parencite{Faugere_1999} --~see the report of \Textcite{MonaganPearce_2015} on this approach.
The theoretical benefits of signature algorithms are diminished by a higher implementation complexity
and a larger output (the signature bases computed by signature algorithms are more constrained than \gbs{}).
More than benchmarks, literature about signature algorithms is turned towards revealing the core ideas
behind F5 and understanding what makes a signature algorithm terminate.
Termination is a very peculiar aspect, not as transparent as termination of Buchberger's algorithm.
Nonetheless, thanks to decisive work by \textcite{HashemiArs_2010,GaoGuanVolny_2010,ArriPerry_2011,EderPerry_2011,GaoVolnyWang_2016}
these goals have been reached in the polynomial case --~see the survey by \textcite{EderFaugere_2017}.

The point of studying signature algorithm may not be the quest of new world records for polynomial system solving,
but rather the understanding of signatures themselves, and what we can extract from them.
Signature bases convey extra information compared to \gbs{}, related to the syzygy module of the input generators. Many ideal-theoretic operations --~intersection, quotient, saturation, Ext modules~--
are related to syzygy modules \parencite{Stillman_1990}, and signature algorithms seem to give an efficient access to them
\parencite{GaoGuanVolny_2010,SunWang_2011,Faugere_2001,EderLairezMohrSafeyElDin_2023}.
Porting these ideas to more general settings is a strong motivation to engage in the study of signatures.
Yet, in my view, the lack of flexibility of the theory of signature algorithms hinders further development, both practical and theoretical. For example, modern implementations for computing \gbs{} make it clear that simultaneous reductions in the F4 style are key towards high performance. Yet, there is no satisfactory description of a signature algorithm with F4-style reduction (\textcite{AlbrechtPerry_2010} do not prove termination and \textcite[\S13]{EderFaugere_2017} are superficial).

This work considers the setting of a module over an algebra over a field, with well-ordered monomials.
This covers many interesting case but excludes some recent developments of signature algorithms which demonstrate the wide applicability of the concept:
signatures in local rings \parencite{LuWangXiaoZhou_2018}, coefficients in Euclidean rings \parencite{EderPfisterPopescu_2017}, principal ideal domains \parencite{FrancisVerron_2020,HofstadlerVerron_2023}, or Tate algebras \parencite{CarusoVacconVerron_2020}, and signatures in a tropical setting \parencite{VacconYokoyama_2017,VacconVerronYokoyama_2018}.

\subparagraph{Contribution} I propose a set of axioms that specifies a context in
which signature algorithms are applicable.
They fit many known settings --~such as solvable algebras \parencite{SunWangMaZhang_2012} and free algebras \parencite{HofstadlerVerron_2022}~--
and some previously unknown settings, such as differential algebras (\S\ref{sec:diff-algebr}).
Very importantly, the axioms describe \emph{modules over a ring}, rather than focusing on the special case of ideals.
While many ideal theoretic constructions (such as ideal intersection) are best understood in terms of modules, they are not addressed in previous works.
The ideas of the most recent frameworks for signature algorithms \parencite{EderPerry_2011,GaoVolnyWang_2016} work smoothly in this axiomatic setting, so many statements will of course be familiar, yet with a wider applicability.

Working with axioms makes some useful ideas emerge.
At least two of them are worth attention.
First, the concept of \emph{prebasis} is introduced to precisely describe admissible inputs for signature algorihms, or, in other words,
to describe what it means for signatures to be consistent.
Previous works all starts by fixing the input,
crafting specific signatures, and developing the theory with respect to these specific signatures.
This is overly restrictive. Moreover, this approach does not highlight the essential properties of signatures, ensuring correctness and termination, among the incidental properties of this specific construction.
This is also problematic when trying to define what signatures and \sbs{} are.
\gbs{} are defined by the equality~$\brack{\lm G} = \lm \brack{G}$, or the confluence of some rewriting system \parencite[e.g.][Definition~5.37]{BeckerWeispfenning_1993},
we do not need to say a \gb{} \emph{of something}.
It is a very desirable definition which should have an equivalent in the signature setting.
The main obstacle here is the definition of signatures, independently of the specific construction that is usually performed for a given input.
This raises an interesting question:
given a set of signatures --~basically anything well-ordered on which act monomials~-- what are the admissible inputs?
Similarly, imagine running some signature basis algorithm from an input~$f_1,\dotsc,f_m$,
and stopping it midway. In this intermediate state, we have polynomials~$g_1,\dotsc,g_r$ with signatures
deriving from the original input by legal operations on polynomials and signatures.
These signatures must be consistent in some sense.
Can we characterize this consistency property without refering to the original input?
This leads to the concept of prebasis, that is
a set of polynomials with signatures which satisfy the fundamental postulates of signatures (elements with equal signatures are substitutable). I prove that if a set of sigpoly pairs is a \gb{} in the module representation, then it is a prebasis (Theorem~\ref{lem:gb-implies-prebasis}).
A concrete application is the reuse of signatures from one computation to another, which is a way to avoid redundant computations.

Second, I introduce \emph{sigtrees} to uncouple the termination criterion from the algorithms themselves. Sigtrees make a “one size fits all” termination criterion.
For Buchberger's algorithm, termination follows from a general principle, Dickson's Lemma,
not from \emph{ad hoc} arguments.
The concept of sigtrees is a tentative to provide such a general argument.
It proves the termination of all known signature algorithms
and also settles positively a conjecture in the classical polynomial setting:
the termination of signature algorithms with out-of-order signature handling and the F5 reductant selection strategy (among all possible reductants, choose the most recent one).
We may blend in an F4-style reduction, the termination argument remains the same.
\looseness=-1

Lastly, as a didactic contribution, I try to emphasize an elemental feature of \sbs{} (or rather, for that matter, \rbs),
putting a clear distinction with \gbs{}.
To check that a given set~$G$ is a \gb{} of an ideal~$I$,
it is enough to check that (1)~$G$ generates~$I$, and (2)~the S-pairs reduce to zero.
Typically (1) will hold {by design} if~$G$ has been constructed from a generating set of~$I$ by usual reduction steps.
Checking~(2) is more difficult and requires arithmetic operations in the base field.
This is typically a costly operation.
In constrast, to check that a given set~$G$ with signatures is a \rb{} of an ideal~$I$,
it is enough to check that (1')~$G$ is a \emph{prebasis} of~$I$, and (2')~the leading monomials and the signatures satisfy some combinatorial property (Theorem~\ref{thm:faugere-criterion}).
The concept of prebasis is introduced in Section~\ref{sec:prebases}, but for the moment, it is enough to say that (1') will hold by design
if~$G$ is obtained by allowed reduction steps from an initial prebasis of~$I$.
The important part is the nature of (2'): it requires no arithmetic operations to be checked, only operations on monomials.
Algorithms for \sbs{} are all about exploiting this combinatorial structure.
This reminds of \emph{staggered linear bases} introduced by \textcite[]{GebauerMoller_1986} to compute \gbs{}, they feature a similar combinatorial structure --~and the link with signatures have recently been investigated by \textcite{HashemiJavanbakht_2021}.


\subparagraph{Plan}
In Section~\ref{sec:gbs}, we define the algebraic structure in which we consider \sbs{}, \emph{monomial modules},
that are vector spaces with a ``leading monomial'' map and an action of a monoid
with some compatibility rules. We also introduce the rewriting system defined by the top reduction.
In Section~\ref{sec:signatures}, we define signatures, \sbs{}, prebases, and \rbs{}. We also state a combinatorial criterion for a set to be a \rb.
Section~\ref{sec:addit-prop-rbs} gathers secondary properties of \rbs{}, such as a precise comparison with \sbs{}, that are not necessary for the next sections.
Section~\ref{sec:algorithms} introduces Noetherian hypotheses, termination arguments and review algorithm \emph{templates}.
Section~\ref{sec:settings} illustrates the axioms by several different settings in which they apply.
\looseness=-1

\subparagraph{Acknowledgment}
I am grateful to Hadrien Brochet and Frédéric Chyzak for a very careful reading and useful comments.
I thank the referees for thoughtful reports.

\section{\gbs{}}
\label{sec:gbs}

Before going to signatures,
we lay down some definitions.
The main ones are the definitions of a \emph{monomial space} --~a vector space with a concept of leading monomial, see Section~\ref{sec:top-reduction}~-- and a \emph{monomial module} --~a monomial space
endowed with a linear action of a (non necessarily commutative) monoid, compatible with leading monomials,
see Section~\ref{sec:monomial-modules}.

In monomial spaces, we develop a (short) theory of \emph{top reduction modulo tail equivalence} (Section~\ref{sec:top-reduction}),
using the terminology of rewriting systems (Section~\ref{sec:rewriting-systems}).
Using rewriting systems to describe the theory of \gbs{} in polynomial rings is done in several textbooks \parencite[e.g.][]{BeckerWeispfenning_1993,Winkler_1996,KreuzerRobbiano_2000,Mora_2005}: in a few words, we say that a polynomial~$f$ can be reduced by a polynomial~$g$, if we can cancel out one of the terms of~$f$ by substracting a scalar multiple of the leading monomial of~$g$.
The context of signatures puts the emphasis on \emph{top reduction} --~the reduction of the leading monomial~--
as opposed to \emph{tail reduction}.
The practice of \gbs{} computation also shows
that tail reduction steps are optional. They are irrelevant as far as termination and correctness is concerned,
to perform them or not is only a matter of performance.
Lastly, tail reduction does not enjoy nice properties.
For example, if~$g$ reduces~$f$ then~$mg$ reduces~$mf$ for any monomial~$m$, in a polynomial setting.
But this implication breaks if~$m$ is a polynomial rather than a monomial,
or if~$f$ and~$g$ lie in a Weyl algebra, unless the reduction is a top reduction.
All of these hints at replacing tail reduction by a more flexible \emph{tail equivalence}
and replace the customary reduction by the top reduction modulo tail equivalence.
This fits the abstract setting of “reduction modulo equivalence” developed by \textcite{Huet_1980}.

\needspace{5\baselineskip}
\subsection{Rewriting systems}
\label{sec:rewriting-systems}

Let~$X$ be a set and~$\myto{1}$ a binary relation on~$X$. ``$x\myto{1} y$'' reads ``$x$ reduces to~$y$''.
Following \textcite{Huet_1980}, we define
the following binary relations:\footnote{Actually Huet denotes either~$\to$ or~$\myto1$ the one-step reduction, which I denote only~$\myto1$,
  and~$\myto*$ the multistep reduction, which I denote~$\to$.}
\begin{itemize}
  \item $x \myto{n} y$, for $n>0$, if there is some~$z\in X$ such that~$x\myto{1} z$ and~$z\myto{n-1}y$;
  \item $x\to y$ if~$x=y$ or~$x\myto{n}y$ for some~$n > 0$, this is the reflexive transitive closure of~$\myto1$;
  \item $x\uparrow y$ if there is some~$z\in X$ such that~$z\to x$ and~$z\to y$;
  \item $x\downarrow y$ if there is some~$z\in X$ such that~$x\to z$ and~$y\to z$.
\end{itemize}

The relation $\myto{1}$ is \emph{Noetherian} if there is no infinite sequence~$x_0\myto1 x_1 \myto1 \dotsb$.
An element~$x\in X$ is \emph{$\myto1$-reduced} if there is no~$y\in X$ such that~$x\myto{1} y$.
If~$x \to y$ and~$y$ is $\myto1$-reduced, then~$y$ is a \emph{normal form} of~$x$.
If~$\myto1$ is Noetherian, then every element has at least one normal form.
The relation $\myto{1}$ is \emph{confluent} if~$x\uparrow y$ implies~$x\downarrow y$ for any~$x, y\in X$.
If~$\myto{1}$ is confluent, then any~$x\in X$ has at most one normal form.


Moreover, given an equivalence relation~$\smile$ on~$X$,
we define:
\begin{itemize}
  \item $x \tup y$ if there are~$z, z' \in X$ such that~$z\smile z'$, $z\to x$ and~$z'\to y$;
  \item $x \tdown y$ if there are~$z, z' \in X$ such that~$z\smile z'$, $x\to z$ and~$y\to z'$;
\end{itemize}
The relation~$\rew$ is \emph{confluent modulo~$\smile$} if $x\tup y$ implies~$x\tdown y$,  for any~$x, y\in X$.


\subsection{Top reduction}
\label{sec:top-reduction}


\begin{definition}[Monomial space, leading monomial, $\eqlt$]
  A \emph{monomial space} over a field~$K$ is a $K$-linear space~$M$ with a basis~$B \subset M$ endowed with a well-order relation~$\leq$.
  The \emph{leading monomial} of~$f\in M$, denoted~$\lm f$, is the $\leq$-maximal element of~$B$ with a nonzero coefficient in~$f$, or~$0$ if~$f = 0$.
  The \emph{set of leading monomials} of~$M$ is defined to be the well-ordered set $B \cup \left\{ 0 \right\}$ where~$0$ is added as the smallest element.

  An equivalence relation~\margindef{$\eqlt$} is defined on~$M$ by~$x\eqlt y$ if~$x=y=0$ or~$\lm(x-y) < \lm x$,
  to be understood as ``$x$ and~$y$ have the same leading term''.
\end{definition}
The convention that~$\lm 0 = 0$ is useful to simplify many statements: being able to write~$\lm f$ without checking that~$f\neq 0$ avoids a case distinction.
From now on, we fix a field~$K$ and a monomial space~$M$ over~$K$. The set of leading monomials of~$M$ is denoted~$\basis$.

\begin{remark}[Equivalent monomial spaces]
  Different choices of a basis~$B$ may lead to equivalent monomial spaces, in the following sense.
  Another well-ordered basis~$B'$ of~$M$ gives an equivalent monomial space if there is an increasing bijection~$\iota : B \to B'$ such that~$\lm' f = \iota(\lm f)$, where~$\lm' f$ is the leading monomial of~$f$ relatively to~$B'$.
  The theory is described only using~$\lm$, not the basis~$B$, so it does not distinguish between equivalent monomial spaces.
  From an axiomatic point of view, we can check that the maps~$\lm$ from~$M$ onto a well ordered set that come from a well-ordered basis, are exactly the maps satisfying
  \begin{enumerate}[label=L\arabic*]
    \item\label{it:lm:zero} $\forall x\in M, \lm x = 0 \Leftrightarrow x = 0$;
    \item\label{it:lm:colin} $\forall x, y\in M, \lm x = \lm y \neq 0 \Leftrightarrow \exists \lambda \in K^\times, \lm(x-\lambda y) < \lm x$.
  \end{enumerate}
  For a given~$f \in M$, we do not define the terms of~$f$, its monomial support, or the coefficient of a monomial in~$f$, because these notions depend on a specific choice of~$B$, which indicates that they are irrelevant in our setting.
\end{remark}

\begin{example}\label{example:w1}
  For polynomial rings, the monomial basis is a very natural choice.
  In a noncommutative setting however, there may be several natural bases.
  For example, the Weyl algebra~$W_1$ generated by two elements~$x$ and~$\partial$ subject to the relation~$\partial x = x \partial + 1$,
  the two natural bases are~$B = \left\{ x^n \partial^m \right\}$ and~$B' = \left\{ \partial^m x^n \right\}$ (with the same possible orderings as the polynomial case). These two bases give equivalent monomial spaces.
\end{example}




\begin{definition}[Top reduction, $\to_E$]\label{def:top-red}
  For any~$E\subseteq M$, the \emph{top reduction}~$\rew_E$ is defined on~$M$ by
  \[ x \rew_E y \Leftrightarrow \lm y < \lm x \text{ and } \exists \lambda \in K^\times, \exists e \in E, y = x - \lambda e. \]
\end{definition}

In other words, $x\rew_E y$ if $y$ is the result of cancelling the leading monomial of~$x$ using a reducer in~$E$.
In this situation, we always have~$\lambda e \eqlt x$.
Since~$x\myto1_E y$ implies~$\lm x < \lm y$, and the set of leading monomials is well-ordered, it is clear that~$\myto1_E$ is Noetherian.

\begin{definition}[Tail equivalence, $\smile_E$, $\tup_E$, $\tdown_E$]
  For any subset~$E\subseteq M$, we define a relation~$\tailred_E$ on~$M$, called \emph{tail equivalence}, defined by
  \[ x \tailred_E y \Leftrightarrow \exists \lambda\in K^\times, \exists e\in E, y = x - \lambda e \text{ and } \lm e < \lm x. \]
  The reflexive transitive closure of~$\tailred_E$ is denoted~$\smile_E$.
  The confluence relations~$\tup_E$ and~$\tdown_E$ are defined using~$\smile_E$.
\end{definition}
Note that~$x\smile_E y$ implies $x \eqlt y$.
The tail equivalence is not a reduction since it is symmetric, it is not defined which side of an equivalence~$x\smile_E y$ is more reduced.

The following statement is a variant, in the setting of monomials spaces, of Buchberger's well known criterion for polynomial ideals.
For~$E\subseteq M$, let~$\brack{E}$ denote the $K$-linear subspace generated by~$E$.

\begin{theorem}[Buchberger's criterion for monomial spaces]\label{thm:linear-buchberger}
  Let~$E$ be a subset of~$M$.
  The following assertions are equivalent:

  (Characterization by leading monomials)\nopagebreak
  \begin{enumerate}[label={B}\arabic*]
    \item\label{item:pivot-basis} $\forall x\in {\brack{E}}, x \neq 0 \Rightarrow \exists e\in E, \lm e = \lm x$.
  \end{enumerate}

  (Characterization by rewriting)
  \begin{enumerate}[label={B}\arabic*, resume]
    \item\label{item:pivot-rewrites} $\forall x\in {\brack{E}}, x \to_E 0$;
  \end{enumerate}

  (Characterizations by confluence properties)
  \begin{enumerate}[label={B}\arabic*, resume]
    \item\label{item:pivot-nf} $\forall x, y\in M, x - y \in \brack{E} \Rightarrow x\tdown_E y$;
    \item\label{item:pivot-confluent} $\to_E$ is confluent modulo~$\smile_E$; 
  \end{enumerate}

  (Characterizations by S-pairs)
  \begin{enumerate}[label={B}\arabic*, resume]
    \item\label{item:pivot-spair} $\forall e, f\in E, \forall \lambda \in K^\times, e \eqlt \lambda f \Rightarrow e-\lambda f \to_E 0$;
    \item\label{item:pivot-regular} $\forall e, f\in E, \forall \lambda \in K^\times, e \eqlt \lambda f \Rightarrow e - \lambda f \in \brack{g\in E\st \lm g<\lm e}$;
  \end{enumerate}
\end{theorem}

\implication{item:pivot-basis}{item:pivot-rewrites}
Let~$x \in \brack{E}$ be a nonzero element
and let~$y$ be a $\to_E$-normal form of~$x$. In particular~$y\in \brack{E}$.
By hypothesis, either~$y=0$, or~$\lm y = \lm e$ for some~$e\in E$. The latter would contradict the irreducility of~$y$, so~$y=0$ and~$x\to_E 0$.

\implication{item:pivot-rewrites}{item:pivot-nf}
Let~$x, y\in M$ such that~$x-y\in \brack{E}$.
By hypothesis, $x-y\to_E 0$.
So there is some~$e\in E$ and~$\lambda\in K^\times$ such that~$x-y \myto1_E x-y-\lambda e \to_E 0$ (unless~$x=y$ but this case is trivial). In particular~$\lambda e \eqlt x-y$.

If~$\lm x > \lm y$, then~$x\myto1_E x-\lambda e$ and~$x-\lambda e - y \in \brack{E}$
and by induction on~$\max(\lm x, \lm y)$, we may assume that~$x-\lambda e\tdown_E y$ and therefore $x\tdown_E y$.
The case~$\lm y > \lm x$ is similar.
If~$\lm x = \lm y$, there is some~$\mu \in  K^\times$ such that~$\lm(x - \mu y) < \lm x$.
There are again two cases. If~$\mu = 1$, that is~$x\eqlt y$,
then the sequence of top-reduction $x-y \myto1 u_1 \myto1_E u_2 \myto1_E \dotsb \myto1_E u_n \myto1 0$
gives a sequence of tail equivalence
\[ x = y + (x-y) \smile_E y + u_1 \smile_E \dotsb \smile_E y + u_n \smile_E y. \]
In particular, $x\tdown_E y$.
If~$\mu \neq 1$, then~$\lm x = \lm y = \lm e$,
so there are reductions~$x \myto1_E x - \kappa e$ and~$y\myto1_E y - \nu e$, for some~$\kappa,\nu\in K^\times$.
By induction on~$\max(\lm x, \lm y)$, we may assume that~$x -  \kappa e \tdown_E y - \nu e$,
which implies~$x \tdown_E y$.

\implication{item:pivot-nf}{item:pivot-confluent}
Let~$x, y\in M$ such that~$x\tup_E y$.
Both~$\myto1_E$ and~$\tailred_E$ preserve equality modulo~$\brack{E}$, so $x-y \in \brack{E}$,
therefore~$x \tdown_E y$, by hypothesis.

\implication{item:pivot-confluent}{item:pivot-spair}
If~$e \eqlt \lambda f$, then~$e \myto1_E e - \lambda f$.
Moreover~$e\myto1_E e-e = 0$. The confluence hypothesis implies that~$e - \lambda f \to_E z$
and~$0 \to z'$ for some~$z, z'\in M$ such that~$z\smile_E z'$.
But~$0$ is reduced and only $\smile_E$-equivalent to itself.
So~$z = 0$ and~$e - \lambda f \to_E 0$.

\implication{item:pivot-spair}{item:pivot-regular}
The rewriting $e - \lambda f \to_E 0$ implies,
by definition of~$\to_E$, that~$e-\lambda f\in \brack{g\in E\st \lm g \leq \lm (e-\lambda f)}$.
Since~$\lm(e - \lambda f) < \lm e$, this gives the claim.

\implication{item:pivot-regular}{item:pivot-basis}
Let~$x\in \brack{E}$
and let~$m\in \basis$ minimal such that~$x\in \brack{e\in E\st \lm e \leq m}$.
We can write~$x = \lambda_1 e_1 + \dotsb + \lambda_r e_r$ with~$e_i \in E$, $\lambda_i \in K$
and~$\lm e_i \leq m$.
We also assume that the~$e_i$ are chosen in such a way that the number~$N$ of indices~$i$ such that~$\lm e_i = m$ is minimal. By minimality of~$m$, we have~$N \geq 1$.

Assume for contradiction that~$\lm x < m$.
In particular~$N \geq 2$ (otherwise, the leading monomials of the $e_i$ cannot cancel to give~$\lm x < m$).
Up to reordering the indices, we may assume that~$m = \lm e_1 = \lm e_2$.
Then~\ref{item:pivot-regular} ensures that~$e_1 - \mu e_2 \in \brack{e\in E\st \lm e < m}$
for some~$\mu \in K^\times$.
We can rewrite~$x$ as
\[ x = \lambda_1 (e_1 - \mu e_2) + (\lambda_2 + \lambda_1 \mu) e_2 + \lambda_3 e_3 + \dotsb + \lambda_r e_r, \]
in contradiction with the minimality of~$N$.
So~$\lm x = m$.
\qed




\bigskip

\begin{definition}[Pivot basis]
A subset~$E\subseteq M$ is a \margindef{pivot basis} if it satisfies the equivalent properties of Theorem~\ref{thm:linear-buchberger}.
\end{definition}

The concept of pivot basis is similar to that of a row echelon form of a matrix.
The following minor lemma, on increasing unions of pivot bases, will be used in Sections~\ref{sec:signature-bases}
and~\ref{sec:rewrite-bases}.

\begin{lemma}
  \label{lem:pivot-incr-union}
  Let~$I$ be a totally ordered set
  and let~$(E_i)_{i\in I}$ be a family of subsets of~$M$.
  If~$E_i \subseteq E_j$ for any~$i,j \in I$ with~$i < j$,
  and if each~$E_i$ is a pivot basis, then~$\cup_{i\in I} E_i$ is a pivot basis.
\end{lemma}

\begin{proof}
  We check the criterion~\ref{item:pivot-spair}.
  Let~$e, f \in \cup_i E_i$ and~$\lambda \in K^\times$ such that~$e \eqlt \lambda f$.
  By definition, $e$ is in some~$E_j$ while~$f$ is in some~$E_k$,
  so both~$e$ and~$f$ are in~$E_{\max(j,k)}$.
  Since~$E_{\max(j,k)}$ is a pivot basis, $e - \lambda f \to 0$ with respect to~$E_{\max(j,k)}$.
  \emph{A fortiori}, it rewrites to~$0$ with respect to~$\cup_i E_i$, which contains~$E_{\max(j,k)}$.
\end{proof}

\subsection{Monomial modules}
\label{sec:monomial-modules}

Recall that a monoid is a set~$A$ with an associative composition
law~$A\times A\to A$ (denoted multiplicatively) which admits an identity element
denoted~$1_A$.

\begin{definition}[Monomial module]
  A \margindef{monomial module} over a monoid~$A$
is a monomial space~$M$ with a linear action of~$A$ on~$M$ (denoted also multiplicatively)
such that:
\begin{enumerate}[label=M\arabic*]
  \item\label{it:momo:lm} $\forall a\in A, \forall f, g \in M, \lm f = \lm g \Rightarrow \lm(af) = \lm(ag)$;
  \item\label{it:momo:ord} $\forall a\in A, \forall f, g\in M, \lm f < \lm g \Rightarrow \lm(af) < \lm(ag)$;
\end{enumerate}
\end{definition}
Note that \ref{it:momo:ord} implies also the following:
\begin{enumerate}[label=M\arabic*, start=3]
  \item\label{it:momo:div} $\forall a\in A, \forall f\in M, \lm(af) \geq \lm f$.
\end{enumerate}
Indeed, if~$\lm(af) < \lm f$, then~$\lm(a^{k+1} f) < \lm(a^k f)$
for any~$k \geq 0$, which would contradicts the well-orderedness of~$\basis$.
Note also that~$M$ is torsionfree: if~$g \neq 0$, then~$ag \neq 0$ for all~$a\in A$,
as a consequence of~\ref{it:momo:ord}.

\begin{definition}[Action on the set of monomials, divisibility]
  The monoid~$A$ acts on the set of monomials~$\basis$ by~$a \lm f \eqdef \lm(af)$.
  A \margindef{divisor} of~$m\in \basis$ is an element~$n \in \basis$ such that~$a n = m$ for some~$a\in A$.
\end{definition}

In the case where~$M$ is a module over a polynomial algebra~$R = K[x_1,\dotsc,x_n]$,
it is natural to choose~$A = \left\{ \smash{x_1^{i_1} \dotsb x_n^{i_n}} \st i_1,\dotsc,i_n \geq 0 \right\}$,
although~$A = R \setminus \left\{ 0 \right\}$ is also a possible choice (with no major theoretical difference).
When~$M$ is a module over a noncommutative ring, the monoid $A = R\setminus \left\{ 0 \right\}$ is a natural choice.
For example, in the Weyl algebra~$W_1$ (see Example~\ref{example:w1}),
the set of monomials~$\left\{ x^n \partial^m \right\}$ is not closed under multiplication (because~$\partial x = x \partial + 1$).
We can also choose~$A$ to be the monoid generated by~$x$ and~$\partial$.
Section~\ref{sec:settings} presents more examples.



\subsection{\gbs{}}
\label{sec:buchberger}

\begin{definition}[\gb]
  Let~$M$ be a monomial space over a monoid~$A$.
  A subset~$G\subseteq M$ is a
\margindef{\gb{}} if $A G$, that is $\left\{ af \st a\in A, f\in G \right\}$,
is a pivot basis.
\end{definition}
It is naturally a key concept, see~\parencite[]{CoxLittleOShea_2015} for an introduction to the topic.
The purpose of signatures is not to replace the concept of \gbs, but rather to give a way to compute them.

\begin{remark}[Singletons]\label{remark:gb-singleton}
  Let~$f$ be a non zero element of~$M$.
  Is~$\left\{ f \right\}$ a \gb{}?
  It will be the case in many practical settings but it is not a consequence of the axioms above.
  A counterexample in a free algebra in two variables is given by \textcite{GreenMoraUfnarovski_1998} (see also Section~\ref{sec:free-algebras}).

  Unfolding the definitions, we see that for every singleton~$\{f\}$ to be a \gb{} is necessary and sufficient that:
  \begin{enumerate}[M1, resume]
    \item\label{it:momo:extra} $\forall f\in M, \forall g\in \brack{a f \st a\in A} \setminus \left\{ 0 \right\}, \exists a\in A, \lm g = \lm(af)$.
  \end{enumerate}
  This holds in most usual settings, and all settings presented in Section~\ref{sec:settings}.
  For example, if~$R = M$ is a polynomial ring, and~$A \subset R$ the monoid of monomials,
  then for any~$g \in \brack{a f \st a\in A}$,
  there is some~$h \in R$ such that~$g = hf$
  and we have~$\lm g = \lm(\lm(h) f)$.
\end{remark}

\section{Signatures}
\label{sec:signatures}

From now on, we fix a monoid~$A$ and two monomial modules over~$A$, denoted~$M$ and~$S$,
with respective sets of monomials denoted~$\basis$ and~$\SIG$.
A signature is an element of~$\SIG$.
We are interested in computing \gbs{} in~$M$ while~$S$ is the module of signatures.

In addition to the axioms for the monomial module~$S$, we also require that
\begin{enumerate}[label=S\arabic*]
  \item\label{it:sig:refines} $\forall a, b \in A, \forall \sigma \in \SIG, \forall m \in \basis, a \sigma = b\sigma \text{ and } \sigma \neq 0 \Rightarrow am = bm$.
  \item\label{it:sig:compat} $\forall a, b\in A, \forall \sigma \in \SIG, \forall m \in \basis, a \sigma \leq b \sigma \text{ and } \sigma \neq 0 \Rightarrow am \leq bm$.
\end{enumerate}
Naturally \ref{it:sig:compat} implies \ref{it:sig:refines}, but we state them separately because \ref{it:sig:compat} will only be useful later in Section~\ref{sec:algorithms} (and specifically in Lemma~\ref{lem:div-implies-dom}) when we will study algorithms for computing \sbs{} and termination issues.
This hypothesis is called \emph{compatibility} by \textcite{GaoVolnyWang_2016} and others.

In concrete situations, we will have a natural module map~$\phi : S \to M$ (that is a~$K$-linear map commuting with the action of~$A$),
but it is not a requirement for the theory.
As \textcite{ArriPerry_2011}, or \textcite{Kambe_2023} more recently, we can define in this context a notion of signature for the elements of~$\phi(S)$
by
\[ \widetilde{\sig}(f) \eqdef \min\left\{ \lm p \in\SIG \st p\in S \text{ and } \phi(p) = f \right\}. \]
This lead however to conceptual difficulties, because the signatures that appear in computations may not coincide with this definition, creating a gap between the theory and tha algorithms.
In the axiomatic approach, we do not define what \emph{is} the signature of elements in~$\phi(S)$.
Instead, we adjunct elements of~$M$ with signatures, and describe the required consistency properties.

\subsection{Signature bases}
\label{sec:signature-bases}

\begin{definition}[Sigpair, sigset, polynomial part, signature]
 A \margindef{sigpair} is an element of~$M\times \SIG$.
 A \margindef{sigset} is a set of sigpairs.
 The first element of a sigpair~$f$, denoted~$f^\natural$, is the \emph{polynomial part} of~$f$ (eventhough~$f$ may not be a polynomial, strictly speaking).
 The second element of a sigpair~$f$, denoted~$\sig f$, is the \margindef{signature} of~$f$.
\end{definition}

\begin{definition}[$\AGlt$, $\AGle$, regular reduction]
For a sigpair~$f$ and some~$a\in A$, we define the sigpair~$a f = (a f^\natural, a \sig f)$.
For notational convenience, we also define a scalar multiplication $\lambda f = (\lambda f^\natural, \sig f)$, for $\lambda \in K^\times$.
For any sigset~$G$, let $A G$ denote the sigset
$AG = \left\{ af \st a\in A, f \in G \right\}$.
For~$\sigma\in\SIG$, let
\[ AG^{\sigma} \eqdef \left\{ \smash{af^\natural} \st a\in A \text{ and } a\sig f = \sigma \right\}, \  AG^{\leq \sigma} \eqdef \cup_{\tau \leq \sigma} AG^\tau \text{ and } AG^{< \sigma} \eqdef \cup_{\tau < \sigma} AG^\tau. \]
They are subsets of~$M$, not sigsets.
Each set~$\AGlt$ defines a reduction rule~$\myto1_{\AGlt}$ (Definition~\ref{def:top-red}), that we denote~$\myto1_G^\sigma$, the \margindef{regular reduction} in signature~$\sigma$.
On~$M\times \SIG$, we define~$f \myto1_G g$ if~$\sig f = \sig g$ and~$f^\natural \myto1_G^{\sig f} g^\natural$.
This is the \emph{regular reduction} of sigpairs.
The tail equivalence relations~$\smile_G^\sigma$ and~$\smile_G$ are defined similarly using~$\smile_{\AGlt}$.
\end{definition}
The reduction relations~$\myto1_G^\sigma$, for any~$\sigma\in \SIG$, are Noetherian, as any top-reduction relation in a monomial space.
So~$\myto1_G$ is also Noetherian, and every sigpair has at least one normal form modulo regular reduction.
The regular reduction~$\to_G$ is the same as the \emph{regular top $\mathfrak s$-reduction} defined by \textcite{EderFaugere_2017}.
In contrast, we will not make use of \emph{singular $\mathfrak s$-reductions} and \emph{tail $\mathfrak s$-reductions}.

\begin{example}[Univariate polynomials]
  Let~$M = S = K[x]$, $\basis = \SIG = \left\{x^n \st n \geq 0\right\} \cup \left\{ 0 \right\}$, with the usual ordering.
  Let~$A = \left\{ x^n \st n \geq 0 \right\}$.
  Let~$G$ be the sigset $\left\{ (x-1, x) \right\}$, made of a single sigpair with polynomial part~$x-1$ and signature~$x$.
  For any~$m \geq 0$, we check that
  \[ \AGlt[x^m] = \left\{ x^k (x-1) \st 0 \leq k < m-1 \right\} \text{ and }  \AGle[x^m] = \left\{ x^k (x-1) \st 0 \leq k \leq m-1 \right\}. \]
  Both are pivot bases. 

  Consider the sigpairs~$f_1 = (x^2, x)$ and~$f_2 = (x^2, x^3)$.
  They have the same polynomial part~$x^2$ but not the same signature.
  The sigpair~$f_1$ is~$\to_G$ reduced.
  Indeed, the only possible reduction to investigate is that of~$f_1$ by~$xg$ (where~$g$ is the unique element of~$G$), but~$\sig(xg) = x^2$, which exceeds~$\sig(f_1) = x$, forbidding the reduction.
  In contrast, we have reductions
  \[ f_2 \myto1_G (x, x^3) \myto1_G (1, x^3), \]
  using~$xg$ and~$g$ successively.
  This exemplifies the additional constraints that signatures impose on reductions, compared to the usual setting without signatures.
\end{example}

The following statements are direct consequences of the axioms for monomial modules.
\begin{lemma}\label{lem:filtration}
  Let~$G$ be a sigset, let~$\sigma\in\SIG$ and~$a\in A$. Then
  \begin{itemize}
    \item for any~$\tau \leq \sigma$, $AG^{\leq \tau} \subseteq AG^{\leq \sigma}$;
    \item for any $f\in AG^{\leq \sigma}$, $af \in AG^{\leq a\sigma}$;
    \item for any $f\in AG^{< \sigma}$, $af \in AG^{< a\sigma}$.
  \end{itemize}
\end{lemma}

Signature-based algorithms for \gbs{} actually compute something more constrained than \gbs.
\begin{definition}[Signature basis]
  A \margindef{\sb{}} is a sigset~$G$ such that for any~$\sigma \in \SIG$
  the set $\AGle$ is a pivot basis.
\end{definition}

Using Theorem~\ref{thm:linear-buchberger}, and a bit of signature manipulation to reduce from~$\AGle$ to~$\AGlt$, we can prove a sigset~$G$ is a \sb{} if and only if regular reduction~$\to_G$ is confluent modulo tail equivalence~$\smile_G$.
Signature bases are a refinement of \gbs{}, in the sense that
forgetting the signatures in a \sb{} gives a \gb.

\begin{lemma}\label{lem:sb-implies-gb}
  If~$G$ is a \sb{}, then~$G^\natural = \left\{ \smash{f^\natural} \st f \in G \right\}$ is a~\gb.
\end{lemma}

\begin{proof}
  The set $A G^\natural$ is the union of all~$A G^{\leq\sigma}$, with~$\sigma \in \SIG$. By construction, $A G^{\leq\sigma} \subseteq AG^{\leq\tau}$ if~$\sigma \leq \tau$. So Lemma~\ref{lem:pivot-incr-union} applies and shows that~$A G^\natural$ is a pivot basis.
\end{proof}

\subsection{Prebases}
\label{sec:prebases}

\begin{definition}[Prebasis]
  A sigset~$G$ is a \margindef{prebasis} if
  \begin{enumerate}[label=P\arabic*,leftmargin=*]
    \item\label{it:prebasis:zero} $AG^0 \subseteq \left\{ 0_M \right\}$;
    \item\label{it:prebasis:colin} $\forall \sigma\in\SIG, \forall f, g\in AG^\sigma, \exists \lambda \in K^\times, f - \lambda g \in \LAGlt$.
  \end{enumerate}
\end{definition}
Equivalently, \ref{it:prebasis:colin} means that any~$f \in AG^\sigma$
generates the quotient space $\LAGle/\LAGlt$ as a~$K$-linear space.
The concept of prebasis embodies the postulate that ``two elements with the same signature are substitutable''.
A prebasis is an admissible input for signature algorithms.

\begin{example}\label{example:trivial-sig}
  A trivial choice for the set of signatures is~$\SIG = \basis$.
  Let~$G$ be a sigset such that~$\sig f = \lm f^\natural$ for any~$f \in G$.
  Then~$G$ is a prebasis if and only if~$G^\natural$ is a \gb{}.
  Indeed, in this case, $\LAGlt = \brack{ \smash{ag^\natural} \st a \in A, g\in G, \lm(ag) < \sigma}$.
  So the condition for being a prebasis is exactly Criterion~\ref{item:pivot-regular} for~$A G^\natural$ to be a pivot basis, that is for~$G^\natural$ to be a \gb{}.
\end{example}

\begin{example}\label{example:trivial-prebasis}
  If~$AG^0 = \varnothing$ and if~$AG^\sigma$ contains exactly zero or one element for any~$\sigma\in \SIG$, then~$G$ is a prebasis.
  This follows directly from the definition.
\end{example}

In the course of computing a \sb, or a \rb{},  as we will see latter,
we will add new elements to a prebasis~$G$ given as input.
Naturally, the construction of new elements must respect both the polynomial part and the signature.
In particular, we want to preserve the prebasis property.
Typically, we construct new elements by regular reduction: for any~$a\in A$ and $g \in G$,
we allow the insertion of any sigpair~$h$ such that~$a g \to_G h$.
In view of using F4-style reductions (Section~\ref{sec:simult-reduct}),
we give a wider definition of allowed extensions of a sigset, that we call \emph{sigsafe extensions}.

\begin{definition}[Sigsafe extension]
  A sigset~$H$ is a \margindef{sigsafe extension} of a sigset~$G$
  if $G\subseteq H$ and for any~$h \in H$, there is some~$f\in AG^{\sig h}$ and some~$\lambda\in K^\times$ such that~$h^\natural \equiv \lambda f^\natural \pmod{\LAGlt[\sig h]}$.
\end{definition}
The problem of computing \sbs{} is more formally stated as ``given a prebasis~$G$, compute a \sb{} that is a sigsafe extension of~$G$''.
For computing a \gb{} of the submodule of~$M$ generated by a set~$G \subseteq M$, we will follow the steps:
first, choose a signature module~$S$,
second, we endow the elements of~$G$ with signatures to turn it into a prebasis;
third, we compute a \sb{} that is a sigsafe extension of~$G$;
four, we remove signatures to obtain a \gb{} (Lemma~\ref{lem:sb-implies-gb}).

\begin{lemma}\label{lem:sigsafe-ext-prebasis}
  Let~$G$ be a prebasis and let~$H$ be a sigsafe extension of~$G$.
  Then:
  \begin{itemize}
    \item $\forall \sigma\in\SIG, \LAGlt = \LAXlt{H}$ and $\LAGle = \LAXle{H}$;
    \item $H$ is a prebasis;
    \item if~$H'$ is a sigsafe extension of~$H$, it is also a sigsafe extension of~$G$.
  \end{itemize}
\end{lemma}
We skip the proof, which is a simple application of Lemma~\ref{lem:filtration}.



Generalizing Example~\ref{example:trivial-sig}, we may construct prebases in~$M$ from a \gb{} in~$S$.

\begin{theorem}\label{lem:gb-implies-prebasis}
  Let~$\phi : S \to M$ be a $K$-linear map commuting with the action of~$A$.
  If~$H \subseteq S$ is a \gb{}, then $\left\{ (\phi(h), \lm h)  \st h \in H \right\} \subseteq M \times \SIG$
  is a prebasis.
\end{theorem}

\begin{proof}
  Let~$G = \left\{ (\phi(h), \lm h) \st h \in H \right\}$.
  We first check~\ref{it:prebasis:zero}.
  Let~$f \in AG^0$.
  By definition, there is some~$h \in H$ and~$a\in A$ such that~$f = \phi(ah)$
  and~$\lm(ah) = 0$. This implies~$ah = 0$, so~$f = 0$.

  As for~\ref{it:prebasis:colin},
  let~$\sigma\in \SIG$, $f,g\in G$
  and~$a, b\in A$ such that~$\sigma = a \sig f=b \sig g$.
  By definition, there are some~$h, k\in H$ such that~$\lm h = \sig f$, $\lm k = \sig g$,
  $f = \phi(h)$ and~$g =\phi(k)$.
  In particular~$\sigma = \lm(ah) = \lm(bk)$, and there is some~$\lambda\in K$ such that~$ah \eqlt \lambda bk$.
  $H$ is a \gb{}, so~$AH$ is a pivot basis, so Criterion~\ref{item:pivot-regular} applies
  and we have~$ah - \lambda bk  = \sum_i m\mu_i c_i l_i$ for some~$\mu_i\in K$, $c_i\in A$ and $l_i \in H$ such that~$\lm(c_i h_i) < \sigma$.
  In particular, $af - \lambda bg = \sum_i \mu_i c_i \phi(l_i)$
  and~$c_i \phi(l_i) \in AG^{< \sigma}$.
\end{proof}

\begin{remark}[Constructing prebases ``for free'']
  \label{remark:prebases-for-free}
  As a special case of Theorem~\ref{lem:gb-implies-prebasis}, we recover the following classical construction
  which underlies all previous work on signature algorithms.
  Given~$g_1,\dotsc,g_r \in M$, we want to find a signature module~$\mathcal{S}$ and signatures~$\sigma_1,\dotsc,\sigma_r$
  such that~$\left\{ (g_i, \sigma_i) \right\}_{1\leq i \leq r}$ is a prebasis.
  This is the first step of the general strategy for computing \gbs{} using signatures.
  The following construction applies when each of the singletons $\left\{ g_i \right\}$ is a \gb{} (this is the common case, see Remark~\ref{remark:gb-singleton}).

  We choose the signature module~$S = M^r \simeq M \otimes K^r$.
  If~$B_M$ is the distinguished basis of the monomial space~$M$,
  we define~$B_S = \left\{ m \otimes e_i \st m\in B_M, 1\leq i \leq r \right\}$
  as the distinguished basis of~$S$, where~$\{ e_1,\dotsc,e_r \}$ denotes the canonical basis of~$K^r$.
  Let~$\SIG = B_S \cup \left\{ 0 \right\}$ denote the set of leading monomials of~$S$.
  There are two natural well-orders on~$\SIG$, called \emph{position-over-term} (POT)
  and \emph{term-over-position} (TOP): $0$ is always the minimal elements, and for the nonzero monomials,
  \begin{description}
    \item[POT] $m \otimes e_i \leq_\SIG n \otimes e_j$ if $i < j$ or~$i=j$ and~$m \leq_\basis n$;
    \item[TOP] $m \otimes e_i \leq_\SIG n \otimes e_j$ if $m <_\basis n$ or~$m=n$ and~$i \leq j$.
  \end{description}
  The monoid~$A$ acts on $S$ by
  $a \cdot (f_1,\dotsc,f_r) = (af_1,\dotsc,af_r)$,
  turning~$S$ into a monomial module over~$A$.
  Moreover, \ref{it:sig:refines} and~\ref{it:sig:compat}
  are satisfied, so~$S$ is a suitable signature module, with either the POT or the TOP ordering.
  Let~$\phi : S \to M$ defined by~$\phi(f_1,\dotsc,f_r) = f_1+\dotsb+f_r$, which commutes with the action of~$A$.

  Since~$\left\{ g_i \right\}$ is a \gb{} in~$M$, for any~$1\leq i\leq r$, it follows easily that~$\left\{ g_i \otimes e_i \right\}$ is a \gb{} in~$S$.
  Moreover, the leading monomials of some~$a g_i \otimes e_i$ and some other~$b g_j \otimes e_j$ can never be equal, unless~$i = j$.
  So it follows that the set~$H = \left\{ g_i \otimes e_i \st 1\leq i \leq r \right\}$ is a \gb{} in~$S$.
  By Theorem~\ref{lem:gb-implies-prebasis},
  this implies that
  \[ G = \left\{ \left( g_i, \lm g_i \otimes e_i \right) \st 1\leq i \leq r \right\} \]
  is a prebasis. And, by construction, $G^\natural = \left\{ g_1,\dotsc,g_r \right\}$.

  This construction shows that, at least under the extra assumption~\ref{it:momo:extra} on~$M$, we have a systematic way to construct
  prebases from arbitrary finite sets of~$M$, enabling the general strategy of using signatures to compute \gbs.
\end{remark}

\begin{example}\label{expl:mora-system}
  Consider the case where~$M = \mathbb{Q}[x, y]$, where the monomial basis of~$M$ is given the degree reverse lexicographic order, with~$x < y$,
  $A = \left\{ x^i y^j \st i, j \geq 0 \right\}$,
  and consider
  \[ g_1 = \underline{x^{2} y^{2}} - 1,\ g_2 = \underline{y^{5}} - x^{2} y, \text{ and } g_3 = \underline{x^{5}} - x y^{2}, \]
  with leading monomial underlined. This is an example of \textcite{Mora_1994}.
  Following the recipe in Remark~\ref{remark:prebases-for-free},
  we consider the signature module~$S = M^3$, with the TOP ordering,
  and we endow the $g_i$ with signatures~$\sig g_i = \lm g_i \otimes e_i$, so we construct the following sigset:
  \[ G = \left\{ \left(  \underline{x^{2} y^{2}} - 1, x^2y^2 \otimes e_1 \right), \left(\underline{y^{5}} - x^{2} y, y^5 \otimes e_2 \right), \left(  \underline{x^{5}} - x y^{2}, x^5 \otimes e_3 \right)\right\}. \]
  In this case, the fact that~$G$ is a prebasis follows from the trivial reason exposed in Example~\ref{example:trivial-prebasis}.

  It is also common to choose the unshifted signature~$\sig g_i = 1 \otimes e_i$.
  It is equally valid to choose~$\sig g_i = m_i \otimes e_i$, for any non zero~$m_i \in \basis$, from the theoretical point of view.
  The choice~$\sig g_i = \lm g_i \otimes e_i$ comes naturally because in the general setting of a monomial module over~$A$, there is no “1”, so we cannot write, in general, $\sig g_i = 1 \otimes e_i$,
  while we can always write~$\sig g_i = \lm g_i \otimes e_i$.
  \textcite[]{EderFaugere_2017} only work in the polynomial case and fix~$\sig g_i = 1\otimes e_i$. However, we can change the term ordering on~$S$ to what they call \emph{lt-pot}, or Schreyer's order \parencite[Definition~2.5]{EderFaugere_2017},
  and recover the behavior of the choice~$\sig g_i = \lm g_i \otimes e_i$ with the TOP order on~$S$.
  \textcite[\S14]{EderFaugere_2017} suggest that this natural choice $\sig g_i = \lm g_i \otimes e_i$, is better for performance than the unshifted signatures. This is exemplified in Figures~\ref{fig:top} and~\ref{fig:unshifted-top}.
\end{example}

\begin{example}[Sum of submodules]
  Let~$G$ and~$H$ be two finite \gbs{} in~$M$.
  Consider the problem of computing a \gb{} $J$ such that~$\brack{AG} + \brack{AH} = \brack{AJ}$.
  We could use, as in Remark~\ref{remark:prebases-for-free},
  the signature module~$M^{\# G + \#H}$ to turn~$G\cup H$ into a prebasis.
  However, this will lead to many useless computations (reductions to zero) because we did not take into account the fact that~$G$ and~$H$ are already \gbs, so all the S-pairs between two elements of~$G$ (resp.~$H$)
  already reduce to~0.

  Instead, we consider the monomial signature module~$S = M^2$
  with the map~$\phi : (u, v) \in S \to u+v \in M$.
  The set
  $B = \left\{ (g, 0) \st g\in G \right\} \cup \left\{ (0, h) \st h \in H \right\}$
  is a \gb{} in~$S$ because the elements~$(g, 0)$ cannot interact with the elements~$(0,h)$.
  By Theorem~\ref{lem:gb-implies-prebasis}, the sigset
  \[ \left\{(\phi(b), \lm b) \st b \in B\right\} = \left\{ (g, \lm g \otimes e_1) \st g\in G \right\} \cup \left\{ (h, \lm h \otimes e_2) \st h \in H \right\}\]
  is a prebasis. We can use it as a starting point to compute a \gb{} of the sum~$\brack{AG} + \brack{AH}$.
  This saves some computations because
  the signatures encode that elements of~$G$ (resp.~$H$) do not need to be reduced with each other.
\end{example}

\subsection{Rewrite bases}
\label{sec:rewrite-bases}


We now introduce \rbs{}. The definition is purely combinatorial, depending only on leading monomials and signatures,
in addition to the prebasis condition.
We will see that a \rb{} is a \sb{} (Corollary~\ref{coro:rb-implies-sb}).
As for pivot bases, prebases, and \gbs{}, being a \sb{} is a matter of subtle arithmetic conditions.
(One cannot change the coefficients of a \sb{} and hope that it remains a \sb.)
Somehow, we can split these conditions into, on the one hand, the prebasis property, and on the other hand, the combinatorial properties of \rbs{}.
The concept was first introduced by \textcite{EderRoune_2013}.
It is simplified here by removing the need for what they call a ``rewrite order''.
So my definition of \rb{} is actually different from theirs, but relates more simply to \sbs{} (Theorem~\ref{prop:groebner-rewrite}).


\begin{definition}[Rewrite basis]
  For~$\sigma\in \SIG$, a sigset~$G$
  is a \margindef{\rb{} at~$\sigma$} if either~$AG^\sigma = \varnothing$, or there is some $\to_G$-reduced element~$f \in AG$ with~$\sig f = \sigma$.
  A sigset~$G$ is a \margindef{\rb{}} if it is a prebasis and a \rb{} at~$\sigma$ for any~$\sigma\in\SIG$.
\end{definition}

\begin{example}[continued]\label{expl:mora-system-2}
  The sigset in Example~\ref{expl:mora-system} is a prebasis but not a \rb.
  The smallest signature at which it is not a \rb{} is~$\sigma = x^2 y^5 \otimes e_2$.
  Indeed,
  \[ AG^\sigma = \left\{ x^2 g_2 \right\} = \left\{ \left(\underline{x^2 y^5} - x^4 y, x^2 y^5 \otimes e_2 \right) \right\}, \]
  and there is a top reduction of~$x^2 g_2$ by~$y^3 g_1$. So~$AG^\sigma$ does not contain any~$\to_G$-reduced element.
  Note that~$x^2 g_2$ does not reduce~$y^3 g_1$ because $\sig(x^2 g_2) > \sig(y^5 g_1)$,
  so~$G$ is a \rb{} at~$x^2 y^5 \otimes e_1$.
  In constrast to the classical setting, the symmetry of critical pairs is broken by the signatures.

  We also check that~$G$ is a \rb{} at any signature~$m\otimes e_1$, for~$m \in \mathcal{M}$.
  These signatures are not empty when~$m$ is a multiple of~$\lm g_1 = x^2y^2$, say~$m = a \lm g_1$.
  There may be a possible reduction when~$m = b\lm g_i$ (with~$i = 2$ or~3),
  but in this case, we have~$a \sig g_1 = a (\lm g_1 \otimes e_1) = m \otimes e_1$
  and, similarly, $b \sig g_i = m \otimes e_i$. The definition of the TOP order, ensures that~$a \sig g_1 < b\sig g_i$, so the reduction is not possible.
\end{example}

The defining property of \rbs{} implies that of \sbs{}.
This is the first aspect of the definition. (See Section~\ref{sec:relation-between-sbs} for more details on the relation between \rbs{} and \sbs{}.)

\begin{proposition}\label{lem:rb-le}
  Let~$G$ be a prebasis and let~$\sigma\in \SIG$
  such that~$G$ is a \rb{} at any signature~$\tau \leq \sigma$.
  Then~$\AGle$ is a pivot basis.
\end{proposition}

\begin{proof}
  For contradiction, assume that~$\AGle$ is not a pivot basis,
  and let~$\tau$ be the smallest signature such that~$\AGle[\tau]$ is not a pivot basis.
  In particular, $\AGlt[\tau]$ is a pivot basis: indeed,
  $\AGlt[\tau]$ is the increasing union~$\cup_{\rho < \tau} \AGle[\rho]$.
  so Lemma~\ref{lem:pivot-incr-union} applies.

  The set~$AG^\tau$ is nonempty, as otherwise~$\AGle[\tau] = \AGlt[\tau]$ and the latter is a pivot basis.
  Since~$G$ is a \rb{} at~$\tau$, there is a~$g \in AG^\tau$ which is $\to_G^\tau$-reduced.
  Since~$\AGle[\tau]$ is not a pivot basis there is a~$f \in \LAGle[\tau]$ such that~$\lm f \neq \lm h$ for any~$h\in \AGle[\tau]$.
  By the prebasis condition \ref{it:prebasis:colin}, there is a~$\lambda\in K$ such that~$f - \lambda g \in \LAGlt[\tau]$.
  And because~$\AGlt[\tau]$ is a pivot basis, by hypothesis, Criterion~\ref{item:pivot-nf} implies that~$f \tdown^\tau_G \lambda g$.
  Since both~$f$ and~$g$ are~$\to^\tau_G$-reduced, we have in fact a tail equivalence~$f\smile^\tau_G g$,
  and therefore~$\lm f = \lm g$, which is a contradiction
\end{proof}

\begin{corollary}\label{coro:rb-implies-sb}
  If~$G$ is a \rb{}, then~$G$ is a~\sb{}.
\end{corollary}

\begin{proof}
  It follows directly from the definitions and Proposition~\ref{lem:rb-le}.
\end{proof}

\begin{corollary}\label{lem:rb}
  Let~$G$ be a prebasis and let~$\sigma\in \SIG$
  such that~$G$ is a \rb{} at any signature $\tau < \sigma$.
  Then~$\AGlt$ is a pivot basis.
\end{corollary}

\begin{proof}
  For contradiction, assume that~$\AGlt$ is not a pivot basis.
  Since $\AGlt$ is the increasing union $\cup_{\tau < \sigma}\AGle[\tau]$,
  there is at least one~$\tau < \sigma$ such that~$\AGle[\tau]$ is not a pivot basis, by Lemma~\ref{lem:pivot-incr-union}.
  This contradicts Proposition~\ref{lem:rb-le}.
\end{proof}

The following statement is an effective form of the prebasis condition, it states that when~$G$ is a \rb, the regular reduction~$\to_G$
is able to witness the prebasis condition: two elements with same signatures have equal $\to_G$-normal forms, up to scaling and tail equivalence.

\begin{corollary}\label{prop:rb}
  Let~$G$ be a prebasis and let~$\sigma\in\SIG$
  such that~$G$ is a \rb{} at any~$\tau < \sigma$.
  For any~$f, g\in AG^\sigma + \LAGlt$, there is some~$\lambda \in K^\times$ such that~$f\tdown_G^\sigma \lambda g$.
\end{corollary}

\begin{proof}
  Let~$f, g\in AG^\sigma + \LAGlt$.  \ref{it:prebasis:colin} implies that there is some~$\lambda\in K^\times$ such that~$f-\lambda g\in \LAGlt$.
  By Corollary~\ref{lem:rb},
  $\AGlt$ is a pivot basis, so
  Criterion~\ref{item:pivot-nf} implies that~$f \tdown_G^\sigma \lambda g$.
\end{proof}

The second aspect of the definition of \rbs{} is the algorithmic content.
Checking if $G$ is a \rb{} at~$\sigma$ involves only manipulations in~$A$, $\basis$ and~$\SIG$, but no operations in the base field~$K$.
Moreover, if~$G$ is not a \rb{} at some~$\sigma$, then it is easy to compute a sigsafe extension of~$G$ which is a \rb{} at~$\sigma$:
simply pick some~$f\in AG$ with~$\sig f = \sigma$, compute a $\to_G$-normal form, and insert the result into~$G$.
This suggests an algorithm schema for computing \rbs{} (Pseudo-algorithm~\ref{algo:abstract-algo}).

\captionsetup[algorithm]{name=Pseudo-algorithm}
\begin{algorithm}[ht]
  \begin{description}
    \item[input] A prebasis~$G$
    \item[output] A sigsafe extension~$H$ of~$G$ that is a \rb{}
  \end{description}
  \begin{pseudo}
    \kw{while} $G$ is not a \rb{} \kw{do} \\+
    pick~$\sigma\in\SIG$ such that~$G$ is not a \rb{} at $\sigma$ \label{line:pseudo:picksig}\\
    pick~$f\in AG$ with~$\sig f = \sigma$ \label{line:pseudo:pickf} \ct{$f$ is called the \emph{reductant}}\\
    $g \gets $ any $\to_G$-normal form of~$f$ \\
    $G \gets G \cup \{ g \}$ \quad \ct{$G$ is now a \rb{} at $\sigma$}\\-
    \kw{return} $G$
  \end{pseudo}
  \caption{Algorithm schema for computing \rbs{}}
  \label{algo:abstract-algo}
\end{algorithm}
\captionsetup[algorithm]{name=Algorithm}

There are two significant difficulties to turn this schema into an actual algorithm.
First, how to check that~$G$ is a \rb{}? And how to pick a signature at which~$G$ is not a \rb? These questions are addressed in Section~\ref{sec:criterion-rbs}.
Second, how to ensure termination? This is addressed, in Section~\ref{sec:algorithms}, under Noetherian hypotheses
and under some restrictions on the choice of~$\sigma$ on line~\ref{line:pseudo:picksig}, or the choice of~$f$ on line~\ref{line:pseudo:pickf}.

\subsection{A criterion for \rbs{}}
\label{sec:criterion-rbs}

There is a criterion (that we will call \emph{Faugère's criterion})
to check that a prebasis is a \rb{}. It plays the same role as Buchberger's criterion plays for \gbs:
reducing a definition that involves infinitely many monomials or signatures to finitely many computations.
However, the analogy between the two criteria is rather thin.
For one, Faugère's criterion is not derived from Buchberger's one and I could not find either a derivation of Buchberger's criterion from Faugère's one.
Moreover, Faugère's criterion only involves combinatorial operations (on leading monomials and signatures) while Buchberger's criterion
involves arithmetic operations through the reductions of S-pairs.
When applying Faugère's criterion, the arithmetic side (that is how the coefficients are relevant)
is hidden in the prebasis hypothesis.

The slogan of signature-based algorithms for \gbs{} is ``process at most one S-pair per signature'',
an algorithmic point of view on the idea that ``two elements with the same signature are substitutable''.
Going one step further, we may ask at which signature we \emph{need} to process a S-pair.
In what follows, the concept of S-pair, inherited from Buchberger's algorithm, fades in favor of a study of the signatures themselves.
This approach is somewhat closer to the concept of J-pairs proposed by \textcite{GaoVolnyWang_2016}: the set of critical signatures that we introduce below is closely related to the set of signatures of J-pairs that need to be handled in the GVW~algorithm.

Our goal here, given a prebasis~$G$, is to define a set of signatures~$\Sigma(G)$ such that it is enough to check that~$G$ is a \rb{} at any signature in~$\Sigma(G)$ to prove that~$G$ is a \rb.
Naturally we want~$\Sigma(G)$ to be as small as possible.
And as soon as we will have introduced Noetherian hypotheses, we will want~$\Sigma(G)$ to be finite and computable in a combinatorial way (that is without arithmetic operations in the base field).

\begin{definition}[Critical set]
  For a sigset~$G$ and a sigpair~$f$,
  the \emph{critical set} of~$f$ modulo~$G$, denoted~\margindef{$\Sigma(f, G)$},
  is the set of all~$\sigma\in \SIG$ such that:
  \begin{enumerate}[label=C\arabic*]
    \item\label{it:sigma:nontr} $\exists a\in A, a \sig f = \sigma$ and~$af$ is not $\to_G$-reduced;
    \item\label{it:sigma:min} $\forall b \in A, \big( b \sig f \text{ is a proper divisor of } \sigma \Rightarrow$ $bf$ is  $\to_G$-reduced$\big)$.
  \end{enumerate}
  The \margindef{critical set} of~$G$, is the set of signatures
  \begin{equation*}
    \label{eq:1}
    \Sigma(G) \eqdef \bigcup_{f\in G} \Sigma(f, G).
  \end{equation*}
\end{definition}
In other words, the condition~\ref{it:sigma:nontr} defines a subset of~$\SIG$
corresponding to the signatures at which a multiples~$af$ is not~$\to_G$-reduced.
This subset is closed under the action of~$A$. Indeed, if there is a reduction~$a f\myto1_G h$, then for any~$b\in A$,
there is also a reduction~$bfa \myto1_G bh$.
Among all these signatures, the condition~\ref{it:sigma:min} retains only the minimal ones for divisibility.
This will be important latter to ensure finiteness.
The important property is the following.

\begin{lemma}\label{lem:critset}
  For any sigset~$G$, any sigpair~$f$, and any~$a\in A$, if~$af$ is not~$\to_G$-reduced, then there is some~$\sigma \in \Sigma(f, G)$ which divides~$a\sig f$.
\end{lemma}

\begin{proof}
  Let~$a' \in A$ such that~$a' \sig f$ divides~$a\sig f$, $a' f$ is not~$\to_G$-reduced, and~$a' \sig f$ is minimal.
  Let~$\sigma = a' \sig f$.
  We check that~$\sigma \in \Sigma(f, G)$.
  Indeed \ref{it:sigma:nontr} follows from the definition.
  For~\ref{it:sigma:min}, let~$b\in A$ such that~$b\sig f$ is a proper divisor of~$\sigma$.
  In particular, $b\sig f$ divides~$a \sig f$ and~$b \sig f < \sigma$.
  By minimality of~$\sigma$, $b f$ is $\to_G$-reduced.
\end{proof}

There is a resemblance with the notion of critical pairs in Buchberger's criterion (see Section~\ref{sec:buchberger})
but also an important difference: critical pairs are elements of~$M$, while the critical set~$\Sigma(f, G)$ only contains signatures, it is a combinatorial content.
Note that~$\Sigma(f, G)$ is included in the union $\cup_{g\in G} \Sigma(f, \left\{ g \right\})$
and~$\Sigma(f, \left\{ g \right\})$ may be thought as the set of signatures of the possible S-pairs between~$f$ and~$g$. In the classical polynomial setting, $\Sigma(f, \left\{ g \right\})$ contains at most one element.
In the general case, $\Sigma(f, \left\{ g \right\})$ can contain zero, one, finitely many or infinitely many elements, see Section~\ref{sec:settings} for examples.

\begin{example}[continued]\label{expl:mora-system-3}
  Consider the sigset~$G$ defined in Example~\ref{expl:mora-system} and developed in Example~\ref{expl:mora-system-2}.
  We compute that
  \[ \Sigma(g_1, G) = \varnothing,\ \Sigma(g_2, G) = \{ x^2 y^5 \otimes e_2 \},\ \Sigma(g_3, G) = \{ x^5 y^2 \otimes e_3 \}. \]
  Note that~$\Sigma(g_3, \left\{ g_1 \right\}) = \{ x^5 y^5\otimes e_3 \}$, reflecting the reduction of~$y^5 g_3$ by~$x^5 g_1$, but this signature disappears in~$\Sigma(g_3, G)$
  because it is divided by~$x^5 y^2 \otimes e_3$, which comes from the reduction of~$y^2 g_3$ by~$x^3 g_1$.
\end{example}

\begin{proposition}
  \label{prop:faugere-criterion-upto}
  Let~$G$ be a prebasis and let~$\sigma\in\SIG$. If~$G$ is a rewrite basis at any signature~$\tau < \sigma$,
  then~$G$ is a \rb{} at~$\sigma$, or~$\sigma\in \Sigma(G)$.
\end{proposition}

\begin{proof}
  Assume that~$G$ is not a \rb{} at~$\sigma$
  and let us prove that~$\sigma\in \Sigma(G)$.
  We may assume that~$AG^\sigma \neq \varnothing$, otherwise~$G$ is a \rb{} at~$\sigma$.
  Let~$a\in A$ and~$f \in G$ such that~$a \sig f = \sigma$.
  We choose~$f$ so that~$\lm(af)$ is smallest.
  By hypothesis, $af$
  is not~$\to_G$-reduced (otherwise~$G$ is a \rb{} at~$\sigma$).
  By Lemma~\ref{lem:critset}, there is a signature~$\tau \in \Sigma(f, G)$ which divides~$\sigma$.
  Let~$b, c\in A$ such that~$b\sig f = \tau$ and~$c \tau = \sigma$.

  If~$\tau = \sigma$, we are done: $\sigma \in \Sigma(f, G)$.
  For contradiction, assume that~$\tau < \sigma$.
  In particular, $G$ is a \rb{} at~$\tau$. So there is some $\to_G^\tau$-reduced~$g \in AG^\tau$.
  By Corollary~\ref{prop:rb}, there is~$\lambda\in K^\times$ such that~$g \tdown_G^\tau \lambda bf$.
  Since~$g$ is~$\to_G^\tau$-reduced and~$bf$ is not, this implies that~$\lm g < \lm  (bf)$, and,
  by~\ref{it:momo:ord}, that
  $\lm(cg) < \lm(cbf)$.
  Moreover, $cb \sig f = a\sig f$, so \ref{it:sig:refines} implies that~$\lm(cbf) = \lm(af)$, and therefore~$\lm(cg) < \lm(af)$.
  This contradicts the minimality of~$\lm(af)$.
\end{proof}

From Proposition~\ref{prop:faugere-criterion-upto}, we easily deduce the following statement.

\begin{theorem}[Faugère's criterion]
  \label{thm:faugere-criterion}
  Let~$G$ be a prebasis. If~$G$ is a rewrite basis at any~$\sigma\in \Sigma(G)$,
  then~$G$ is a \rb{}.
\end{theorem}

\section{Additional properties of \rbs}
\label{sec:addit-prop-rbs}

This section gathers some properties of \rbs{} that are not central, and not
used in the next sections, but that connect to previous works.

\subsection{Relation between \sbs{} and \rbs}
\label{sec:relation-between-sbs}

Corollary~\ref{coro:rb-implies-sb} shows that \rbs{} are \sbs{}.
With two competing definitions, it is worth studying more precisely the relation between them.

We first introduce a classification of signatures.
Let~$G$ be a prebasis.
For any~$\sigma\in\SIG$, either~$AG^\sigma = \varnothing$, this is a trivial case, or any element of~$AG^\sigma$ generates the quotient~$\LAGle/\LAGlt$.
In the latter case, either  every element of~$AG^\sigma$ is in~$\LAGlt$, if the quotient is zero-dimensional,
or no element of~$AG^\sigma$ is in~$\LAGlt$, if the quotient is one-dimensional.
This leaves the following categories.
A signature~$\sigma\in\SIG$ is:
\begin{itemize}
  \item an \emph{empty signature} if~$AG^\sigma = \varnothing$;
  \item a \emph{nonempty signature} if~$AG^\sigma \neq \varnothing$.
\end{itemize}
Moreover, a nonempty signature is:
\begin{itemize}
  \item a \emph{regular signature} if $AG^\sigma \cap \LAGlt = \varnothing$;
  \item a \emph{syzygy signature} if $AG^\sigma \subseteq \LAGlt$.
\end{itemize}
A nonempty signature is either regular or syzygy, as long as~$G$ is a prebasis.
This classification is relative to the sigset~$G$, but we check easily that it remains unchanged
under sigsafe extensions.

\begin{example}[continued]
  Consider again the prebasis~$G$ from Example~\ref{expl:mora-system}.
  We check easily that:
  \begin{itemize}
    \item $1\otimes e_1$ is an empty signature, because it is not multiple of any of the signatures in~$G$.
    \item $\sigma  = x^5 y^5 \otimes e_3$ is a (nonempty) syzygy signature. Indeed, $AG^\sigma = \left\{ y^5 \smash{g_3^\natural} \right\}$ and  $y^5 g_3 \to_G 0$, using the reducer~$x^5 g_2$.
          So~$y^5 g_3^\natural \in \LAGlt$, and therefore~$AG^\sigma \subseteq \LAGlt$.
    \item $\sigma = x^2 y^2 \otimes e_1$ is a (nonempty) regular signature. Indeed~$\AGlt = \varnothing$, so~$\LAGlt = 0$ while~$AG^\sigma$ contains the nonzero element~$g_1^\natural$.
  \end{itemize}
\end{example}

\begin{proposition}\label{prop:groebner-rewrite}
  Let~$G\subseteq M$ be a prebasis.
  $G$ is a \sb{} if and only if\/
  $G$ is a rewrite basis at any regular signature.
\end{proposition}

\begin{proof}
  Assume first that~$G$ is a \sb{}. Let~$\sigma\in \SIG$ be a regular signature
  and let us prove that~$G$ is a rewrite basis at~$\sigma$. Because~$\sigma$ is
  regular there is some~$f \in AG$ with~$\sig f = \sigma$.
  Let~$v$ be a $\to_G$-normal form of~$f$,
  with respect to~$G$.

  The signature~$\sigma$ is regular, so~$AG^\sigma \cap \LAGlt = \varnothing$.
  In particular~$f^\natural$ is not in~$\LAGlt$, and thus~$v^\natural$ is not zero.
  Because~$\AGle$ is a pivot basis, by definition of a \sb, $v^\natural$ is reducible with respect to~$\AGle$.
  So there is some~$g\in AG$ such that~$\sig g \leq \sigma$ and~$\lm g = \lm v$.
  But~$v$ is $\to_G$-reduced so~$\sig g = \sigma$.
  Moreover~$\lm v = \lm g$, so~$g$ is also~$\to_G$-reduced.
  So~$G$ is a \rb{} at~$\sigma$.

  The converse follows from the same argument used in the proof of Corollary~\ref{coro:rb-implies-sb}.
\end{proof}

The only property
that a \sb{} misses to be a \rb, is an explicit marking of syzygy signatures
by sigpairs with polynomial parts equal to zero.
The data of syzygy signatures is a by-product of all known signature-based algorithms.
So actually, they compute \rbs, not only \sbs.
The following statement establishes an equivalence which
does not hold for the original definition of \rbs{} by \textcite[\S 3.2]{EderRoune_2013},
only the “\rb{} $\Rightarrow$ \sb” implication holds for this definition.\footnote{With the simplified definition, there must be at least one $\to_G$-reduced element per signature. With the original signature, one specific element must be $\to_G$-reduced.}
This is the main motivation for the simplified definition.

\begin{corollary}\label{prop:groebner-rewrite-syz}
  Let~$G \subseteq M$ be a prebasis.
  $G$ is a \rb{} if and only if the following hold:
  \begin{itemize}
    \item $G$ is a \sb{};
    \item  for any syzygy signature~$\sigma$, there is some~$g\in G$ such that~$\sig g$ divides~$\sigma$ and~$g^\natural = 0$.
  \end{itemize}
\end{corollary}


\subsection{Minimal elements in \rbs{}}

We first introduce a binary relation on the set of sigpairs.
\begin{definition}[Domination relation, $\sqsubseteq$]\label{def:dom}
  We say that~\emph{$g$ dominates~$f$}\marginindex{domination}, and denote it~\margindef{$g \sqsubseteq f$}, if one of the following holds:
  \begin{enumerate}[label=D\arabic*]
    \item\label{it:dom:non-min} $\exists a \in A, a  \sig g = \sig f$ and $a \lm g \leq \lm f$;
    \item\label{it:dom:non-topred} $\exists a \in A, a  \sig g < \sig f$ and $a \lm g = \lm f \neq 0$.
  \end{enumerate}
  A sigpair $f$ is \margindef{dominant} in a sigset~$G$ if~$f \in G$ and for any $g \in G$ such that~$g \sqsubseteq f$,
  we also have~$f \sqsubseteq g$.
\end{definition}
Note that the domination relation may not be transitive, although both~\ref{it:dom:non-min} and~\ref{it:dom:non-topred}, considered separately,
define a transitive relation.
Note also that~\ref{it:dom:non-min} is the covering relation defined by \textcite[p.~454]{GaoVolnyWang_2016}.

The elements of a sigset that are strictly dominated are useless in a \rb{}.
It is important to understand why.
The condition~\ref{it:dom:non-topred} means that~$ag$ can be used to top-reduce~$f$, so~$f$ will never help any sigset containing also~$g$ to be a rewrite basis.
The interpretation of the condition~\ref{it:dom:non-min} splits into two cases.
First, when~$a\lm g = \lm f$, then~$f$ will not help because~$ag$ can serve just as well in any situation where~$f$ would serve.
When~$a\lm g < \lm f$, Corollary~\ref{prop:rb} proves that~$f$ will never be reduced in a rewrite basis containing~$g$.



\begin{theorem}
  Let~$G$ be a prebasis and~$H$ be a sigsafe extension
  such that  every element of~$H$ is dominated by an element of~$G$.
  Let~$\sigma$ be a signature such that~$H$ and~$G$ are \rbs{} at any signature~$\tau < \sigma$.
  Then $H$ is a \rb{} at~$\sigma$ if and only if~$G$ is a \rb{} at~$\sigma$.
\end{theorem}

\begin{proof}
  A sigsafe extension of a \rb{} is a \rb, so one implication is clear.
  Conversely, assume that~$H$ is a \rb{} at~$\sigma$.
  Because~$H$ is a \rb{} at~$\sigma$, there are
  $b\in A$ and~$f \in H$ such that~$b \sig f = \sigma$ and~$bf$ is $\to_H$-reduced (and thus $\to_G$-reduced too).
  By hypothesis, there is some~$g \in G$ such that~$g \sqsubseteq f$.
  Since~$bf$ is $\to_G$-reduced, \ref{it:dom:non-topred} cannot hold, so~\ref{it:dom:non-min} does:
  there is some~$a\in A$ such that~$a \sig g = \sig f$ and~$a\lm g \leq \lm f$.

  Since~$H$ is a sigsafe extension of~$G$, $f \in AG^\sigma + \LAGlt$ (maybe after a scalar multiplication), by definition.
  By Corollary~\ref{prop:rb},
  there is some~$\lambda \in K^\times$ such that
  $bf \tdown_G \lambda bag$.
  Since~$bf$ is~$\to_G$-reduced, this implies $\lm(bf) \leq \lm(bag)$.
  Combining with the previous inequality, we obtain that~$\lm(bag) = \lm(bf)$. So~$bag$,
  which has same leading monomial and signature as~$bf$, is $\to_G$-reduced and thus~$G$ is a \rb{} at~$\sigma$.
\end{proof}

Combining with Theorem~\ref{thm:faugere-criterion}, we obtain the following corollary
which may be used to reduce the number of signatures to consider when computing a \rb{}.
It allows, during the computation of a \rb, to consider only the critical signatures relative to the dominant elements, while retaining the nondominant elements for computing the reductions.

\begin{corollary}
  Let~$G$ be a prebasis and~$H$ be a sigsafe extension
  such that  every element of~$H$ is dominated by an element of~$G$.
  If $H$ is a \rb{} at any~$\sigma\in \Sigma(G)$, then~$G$ and~$H$ are \rbs.
\end{corollary}

\subsection{Syzygies}

When a rewrite basis comes from a \gb{} in the signature module through a map~$\phi : S\to M$ (see Section~\ref{sec:prebases}),
the syzygy signatures have an interpretation in terms of the kernel of~$\phi$.
This is an important feature of \rbs{} that can be exploited to compute efficiently colon ideals and saturations \parencite[]{GaoGuanVolny_2010,EderLairezMohrSafeyElDin_2023}.

\begin{proposition}
  Let~$\phi : S\to M$ be a linear map commuting with the action of~$A$,
  let~$H\subseteq S$ be a \gb{},
  let~$G = \left\{ (\phi(h), \lm h) \st h\in H \right\}$,
  and~$J = \left\{ h\in H \st \phi(h) = 0 \right\}$.
  If~$G$ is a \rb{}, then~$J$ is a \gb{}
  and~$\brack{AJ} = \ker \phi \cap \brack{AH}$.
\end{proposition}

\begin{proof}
  It is clear that~$\brack{AJ} \subseteq \ker \phi \cap \brack{AH}$.
  Let~$h \in \ker \phi \cap \brack{AH}$,
  let~$\sigma = \lm h$
  and let us prove that there is some~$k \in AJ$ such that~$\lm h = \lm k$.
  This will prove at the same time that $J$ is a \gb{}, using Criterion~\ref{item:pivot-basis},
  and that~$\brack{AJ} = \ker \phi \cap \brack{AH}$.

  Because~$H$ is a \gb{},
  we can decompose~$h$ as~$\lambda p + q$, with $p\in AH^\sigma$, $q \in \LAXlt{H}$ and~$\lambda\in K^\times$ (using the first reduction step of the reduction given by Criterion~\ref{item:pivot-rewrites}).
  In particular~$\lambda^{-1} \phi(h) \in AG^\sigma + \LAGlt$.
  Since~$G$ is a \rb{} at $\sigma$, there is some~$a \in A$ and~$g\in G$ such that
  $ag$ is $\to_G$-reduced and~$a\sig g = \sigma$.
  By Corollary~\ref{prop:rb},
  we have~$\phi(h) \tdown_G^\sigma \mu ag^\natural$ for some~$\mu\in K^\times$.
  But~$\phi(h) = 0$ and~$ag^\natural$ is $\to_G^\sigma$-reduced, so~$ag^\natural = 0$
  and therefore~$g^\natural = 0$.
  By definition of~$H$, $g = (\phi(k), \lm k)$ for some~$k\in H$.
  And since~$g^\natural = 0$, we have~$k \in J$.
  In particular, $\lm h = \sigma = \lm(ak)$.
\end{proof}

\section{Algorithm templates}
\label{sec:algorithms}

In all this section we assume that~$M$ and~$S$ are Noetherian monomial modules, which we define in Section~\ref{sec:noeth-monom-modul}.
(Note that this is unrelated to the property that the regular reduction~$\to_G$ is Noetherian.)
This will imply the finiteness of the critical set~$\Sigma(G)$ of finite sigsets~$G$
as well as the existence of finite sigsafe extensions that are \rbs{}, for any sigsets.

As it will become clear, there is not a single algorithm for computing \rbs{}.
There are many possible variants, some major, such as F5 selection strategy or F4-style reduction, and some minor.
There are also many possible ways to combine them.
More than to prescribe some algorithms, the goal of this section is to highlight design principles.

Section~\ref{sec:noeth-monom-modul} introduces the Noetherian hypotheses.
Section~\ref{sec:proc-sign-order} studies an algorithm where signatures are processed \emph{in order},
that is when a signature is always processed after any smaller signatures.
This is a natural setting, yielding simple proofs of termination, but it does not fit all situations.
Section~\ref{sec:minim-lead-monom} studies the idea of minimizing the leading monomial of the reductant,
in the style of \textcite{ArriPerry_2011} and \textcite{SunWang_2011}.
Again, it leads to rather simple proofs of termination, but it leaves aside other reductant selection strategy, such as the original F5 strategy.

To study algorithms where the signatures may be processed \emph{out of order}
and the reductant selected (almost) freely, Section~\ref{sec:well-formed-sigtrees}
introduces \emph{sigtrees}. It is a tree whose nodes are the elements of the \rb{} being computed,
and $g$ is a child of~$f$ if was obtained from a reduction by a multiple of $f$.
Under mild hypotheses, sigtrees are finite (Theorem~\ref{thm:sigtree-term}), giving a very useful termination criterion.
This criterion is put into practice in Section~\ref{sec:f5-reduct-select}, to study the F5 selection strategy
with out-of-order signature processing, in Section~\ref{sec:expl-manag-crit},
to study the most general selection strategy, according to the sigtree criterion,
and in Section~\ref{sec:simult-reduct} to study simultaneous reduction in the F4 style.

\subsection{Noetherian monomial modules}
\label{sec:noeth-monom-modul}

A partial order~$\trianglelefteq$ on a set~$X$ is a \emph{well partial order} (or \emph{wpo})
if for any sequence~$(x_i)_{i \geq 0}$ in~$X$, there are some~$i < j$ such that~$x_i \trianglelefteq x_j$.
A subset~$T$ of a partially ordered set~$X$ is \emph{closed} if~$a\trianglelefteq b$ and~$a\in T$ imply~$b\in T$.
Wpos have several equivalent characterizations.

\begin{lemma}[{\cite[Theorem~2.1]{Higman_1952}}]
  \label{lem:noetherian}
  Let~$X$ be a set with a partial order~$\trianglelefteq$.
  The following assertions are equivalent:
  \begin{enumerate}[label=N\arabic*]
    \item any sequence $T_0 \subseteq T_1 \subseteq \dotsc$ of closed subsets of~$X$ stabilizes;
    \item\label{it:wpo:pair} for any sequence~$(x_i)_{i\geq 0}$ in~$\SIG$, there are some~$i<j$ such that~$x_i \trianglelefteq x_j$ (i.e. $\trianglelefteq$ is a wpo);
    \item\label{it:wpo:extract} for any sequence~$(x_i)_{i \geq 0}$ in~$\SIG$, there is a subsequence~$(x_j)_{j \geq 0}$ such that~$x_j \trianglelefteq x_{j+1}$ for any~$j \geq 0$;
    \item\label{it:wpo:fpb} for any closed set~$T\subseteq X$, there is finite set~$B$ such that~$T = \left\{ x\in X\st \exists b\in B, b \trianglelefteq x \right\}$.
  \end{enumerate}
\end{lemma}

A monomial set~$\basis$ of a monomial module~$M$ is partially ordered by divisibility,
an order that we will denote~$\trianglelefteq$, not to be confused with the total order~$\leq$.
Namely, $m\trianglelefteq n$ if there is some~$a\in A$ such that~$am = n$.
Nonetheless, if~$a\trianglelefteq b$ then~$a\leq b$, by~\ref{it:momo:div}.
\begin{definition}[Noetherian monomial space]
  A monomial space~$M$ is \margindef{Noetherian}
  if~$\trianglelefteq$ is a wpo.
\end{definition}

From now on, we assume that the monomial spaces~$M$ and~$S$ (the signature module)
are Noetherian. The first interesting consequence is the finiteness of the critical set~$\Sigma(G)$ for a given finite sigset~$G$.

\begin{lemma}\label{lem:finite-crit-sig}
  Let~$G$ be a finite sigset. If~$S$ is Noetherian, then~$\Sigma(G)$ is finite.
\end{lemma}

\begin{proof}
  Let~$f\in G$.
  By definition, $\Sigma(f, G)$ is the set of $\trianglelefteq$-minimal elements of some closed subset of~$\SIG$.
  By Criterion~\ref{it:wpo:fpb}, it is finite.
\end{proof}

The termination arguments will not follow from the Noetherianity of~$M$ or~$S$ alone,
but in conjunction.
More precisely, in~$\basis \times \SIG$  we define~$(m, \sigma) \trianglelefteq (n, \tau)$ if~$m \trianglelefteq n$ and~$\sigma \trianglelefteq \tau$.
In other words, $(m, \sigma) \trianglelefteq (n, \tau)$ if there are~$a, b\in A$ such that~$am = n$ and~$b \sigma = \tau$.
Let us insist that $a$ and~$b$ may not be equal.

\begin{lemma}\label{lem:MSwpo}
  If~$M$ and~$S$ are Noetherian monomial modules, then~$\trianglelefteq$ is a wpo on~$\basis \times \SIG$.
\end{lemma}
\begin{proof}
  Let~$\big((m_i, \sigma_i)\big)_{i \geq 0}$ be an infinite sequence in~$\basis\times \SIG$.
  By Criterion~\ref{it:wpo:extract},
  we may assume, up to extracting a subsequence, that~$m_i \trianglelefteq m_{i+1}$.
  Similarly, we may assume, up to extracting a subsubsequence, that~$\sigma_i \trianglelefteq \sigma_{i+1}$.
  So~$\trianglelefteq$ on~$\basis\times\SIG$ satisfies Criterion~\ref{it:wpo:extract}.
\end{proof}

The following statement relates $\trianglelefteq$ with the domination relation~$\sqsubseteq$ (Definition~\ref{def:dom}).

\begin{lemma}\label{lem:div-implies-dom}
  For any sigpairs~$f$ and~$g$, if~$(\sig g, \lm g) \trianglelefteq (\sig f, \lm f)$
  then~$g \sqsubseteq f$.
\end{lemma}

\begin{proof}
  Let~$a, b\in A$ such that~$a\sig g=\sig f$ and~$b\lm g = \lm f$.
  If~$b \sig g < \sig f$, then \ref{it:dom:non-topred} holds.
  Otherwise, if~$a \sig g = \sig f \leq b \sig g$, then \ref{it:sig:compat} implies that~$a \lm g \leq b \lm g = \lm f$ (so \ref{it:dom:non-min} holds), unless~$\sig g = 0$. In this last case, we have~$\sig f = b\sig g = 0$ and~$b\lm g = \lm f$, so~\ref{it:dom:non-min} also holds.
\end{proof}

The following statement will underlie all the termination proofs.
It is an analogue of Dickson's Lemma for sigpairs. However, we will see that this statement may not apply directly.
The relation $\trianglelefteq$ on~$\basis\times \SIG$, the domination relation~$\sqsubseteq$
and Lemma~\ref{lem:div-implies-dom} appeared
first in the work of
\textcite{ArriPerry_2011,ArriPerry_2017} and they have been used several times since then
\parencite{EderPerry_2011,RouneStillman_2012,GaoVolnyWang_2016}.

\begin{proposition}[Dickson's Lemma for sigpairs]\label{prop:dom-noetherian}
  For any infinite sequence~$(f_i)_{i \geq 0}$ of sigpairs, there are indices~$i< j$ such that~$f_i \sqsubseteq f_j$.
\end{proposition}

\begin{proof}
  It is a direct corollary of Lemma~\ref{lem:MSwpo}, Criterion~\ref{it:wpo:pair} and Lemma~\ref{lem:div-implies-dom}.
\end{proof}

\subsection{Processing signatures in order}
\label{sec:proc-sign-order}

By Proposition~\ref{prop:faugere-criterion-upto},
we can compute the smallest signature at which a given prebasis~$G$ is not a \rb{}: it must be an element of the critical set~$\Sigma(G)$, which is finite by Lemma~\ref{lem:finite-crit-sig}.
This signature has many good properties induced by Corollary~\ref{prop:rb}, and in particular we deduce the following one.

\begin{proposition}
  \label{prop:minimal-sig-implies-dominant}
  Let~$G$ be a prebasis and let~$\sigma$ such that~$G$ is a \rb{} at any~$\tau < \sigma$.
  Let~$f \in AG$ with~$\sig f = \sigma$
  and let~$h$ be any~$\to_G$-normal form of~$f$.
  Then either~$G$ is a \rb{} at~$\sigma$, or $g \not\sqsubseteq h$ for any~$g \in G$.
\end{proposition}

\begin{proof}
  Assume that there is some~$g\in G$ such that~$g\sqsubseteq h$.
  Domination condition~\ref{it:dom:non-topred} is ruled out because~$h$ is $\to_G$-reduced.
  Therefore~\ref{it:dom:non-min} holds: there is some~$a\in A$ such that~$a \sig g = \sigma$ and~$\lm(ag) \leq \lm h$.
  By Corollary~\ref{prop:rb},
  $f \tdown_G \lambda ag$ for some~$\lambda \in K^\times$.
  By confluence, we also have~$h \tdown_G \lambda ag$.
  Since~$h$ is $\to_G$-reduced, this implies that~$\lm h \leq \lm(ag)$.
  Combining with the condition~\ref{it:dom:non-min}, we obtain that~$\lm(ag) = \lm h$ and therefore that~$ag$ is also $\to_G$-reduced.
  So~$G$ is a \rb{} at~$\sigma$.
\end{proof}

This leads to Algorithm~\ref{algo:rb-min-sig}.
There is no restriction whatsoever on the choice of the reductant on line~\ref{line:minsig:pickf},
they all reduce to the same sigpair, up to scaling and tail equivalence~$\smile_G$ (Corollary~\ref{prop:rb}).

\begin{algorithm}
  \begin{description}
    \item[input] A finite prebasis~$G$
    \item[output] A finite sigsafe extension of~$G$ which is a \rb{}
  \end{description}
  \begin{pseudo}
    \kw{while} $G$ is not a \rb{} at all $\sigma\in \Sigma(G)$ \kw{do} \label{line:minsig:while}\\+
    $\sigma \gets \min \left\{ \sigma \in \Sigma(G) \st G \text{ is not a \rb{} at $\sigma$} \right\}$ \\
    pick any~$f\in AG$ with~$\sig f = \sigma$ \label{line:minsig:pickf}\\
    $g \gets $ any $\to_G$-normal form of~$f$ \label{line:minsig:red}\\
    $G \gets G \cup \left\{ g \right\}$ \label{line:minsig:insert}\\-
    \kw{return} $G$
  \end{pseudo}
  \caption{Computation of a \rb{} handling signatures in increasing order}
  \label{algo:rb-min-sig}
\end{algorithm}

\begin{theorem}\label{thm:algo-minsig}
  Algorithm~\ref{algo:rb-min-sig} is correct and terminates.
\end{theorem}

\begin{proof}
  Correction follows from Theorem~\ref{thm:faugere-criterion}.
  For contradiction, assume that the algorithm does not terminate for some input.
  Let~$g_1,g_2,\dotsc$ be the sigpairs that are inserted to~$G$ on line~\ref{line:minsig:insert} on each iteration.
  By Proposition~\ref{prop:dom-noetherian}, there are some indice~$i< j$ such that~$g_i \sqsubseteq g_j$.
  Proposition~\ref{prop:faugere-criterion-upto} implies that when~$g_j$ is picked, $G$ is a \rb{} at any signature $< \sig g_j$.
  So Proposition~\ref{prop:minimal-sig-implies-dominant} implies that~$g_i \not\sqsubseteq g_j$, which is a contradiction.
\end{proof}

This algorithm is close in essence to the original F5 algorithm \parencite{Faugere_2002}
and more generally to the RB algorithm of \textcite{EderPerry_2011}.
The notion of critical set and the notation~$\Sigma(G)$ greatly simplify the presentation of the algorithm,
but it hides combinatorial computations.
For example, how to update~$\Sigma(G)$ after inserting a new element?
How to find the next signature to handle? How to check the halting condition?
These questions are addressed in Section~\ref{sec:expl-manag-crit}.

\begin{example}[continued]\label{expl:mora-system-4}
  Let us apply Algorithm~\ref{algo:rb-min-sig} to Mora's system (Examples~\ref{expl:mora-system}, \ref{expl:mora-system-2} and~\ref{expl:mora-system-3}).
  Let~$G_{i+3}$ denote the value of~$G$ at the end of the $i$th iteration. (So that~$G_3$ is the input, starting the counter at~3 because the input contains already 3~elements.)
  At the start of the algorithm, we have~$\Sigma(G_3) = \left\{ x^2 y^5 \otimes e_2, x^5 y^2 \otimes e_3 \right\}$.
  The minimal element of~$\Sigma(G_3)$ is~$\sigma_4 = x^2 y^5 \otimes e_2$ and~$G$ is not a \rb{} at~$\sigma_4$.
  We pick the reductant~$x^2 g_2$ (only possible choice) and using~$y^3 g_1$, we compute the reduction
  \[ x^2 g_2 \to_G  -\underline{x^4 y} + y^3, \text{in signature } x^2y^5 \otimes e_2 \]
  We have a new basis element~$g_4 = (\underline{-x^4 y} + y^3, \sigma_4)$ to obtain~$G_4$.
  The set~$\Sigma(g_4, G_4)$ gives two new elements of~$\Sigma(G_4)$:
  \[ \Sigma(G_4) = \Sigma(G_3) \cup \left\{ x^2 y^6 \otimes e_2, x^3 y^5 \otimes e_2 \right\}, \]
  which gives two new elements in~$\Sigma(G_4)$.
  The minimal element of~$\Sigma(G_4)$ at which~$G_4$ is not a \rb{} is~$\sigma_5 = x^5 y^2 \otimes e_3$.
  We pick the reductant~$y^2 g_3$ (only choice) and using~$x^3 g_1$, we compute the reduction
  \[ y^2 g_3 \to_G  -\underline{x y^4} + x^3, \text{in signature } x^5y^2 \otimes e_3. \]
  We have a new basis element~$g_5 = (\underline{-x y^4} + x^3, \sigma_5)$ to obtain~$G_5$.
  We have two new elements in the critical set~$\Sigma(G_5)$:
  \[ \Sigma(G_5) = \Sigma(G_4) \cup \left\{ x^6 y^2 \otimes e_2, x^5 y^3 \otimes e_2 \right\}. \]
  The next signature is~$\sigma_6 = x^2 y^6 \otimes e_2$.
  We pick the reductant~$yg_4$, and we have
  \[ yg_4 \to_{G_5} \underline{y^{4}} - x^{2}, \text{ in signature~$x^2 y^6 \otimes e_2$}, \]
  leading to a new element~$g_6$.
  Another possible choice of reductant is~$x^2y g_2$, which would lead to the same~$g_6$.
  From the computational point of view, it is clear that~$y g_4$ is a better choice than~$x^2 y g_2$ because~$g_4$ was already obtained by reducing~$x^2 g_2$,
  so there will be less work to reduce~$yg_4$ than to reduce~$x^2 y g_2$.
  When the signatures are not process in order, the choice of the reductant have a theoretical importance discussed in the following sections.

  The process goes on similarly.
  We can represent the output of Algorihm~\ref{algo:rb-min-sig} in the form of trees (which are instances of the concept of \emph{well-formed sigtree} introduced in Section~\ref{sec:well-formed-sigtrees}).
  We say that~$h\in G$ is the parent of~$g\in G$ if~$g$ is obtained, on line~\ref{line:minsig:red}, from the reduction of~$f = ah$
  for some~$a\in A$.
  To display the tree, we show only the leading monomials of the sigpairs, the iteration at which they have been inserted,
  and an edge~$h\to g$ is labelled by the element~$a \in A$ defined above.
  For the input discussed above, we obtain Figure~\ref{fig:mora}.
  The data displayed in this tree (leading monomials and signatures) is enough to certify that the output is a \rb{}, using Theorem~\ref{thm:faugere-criterion}.
\end{example}

\begin{figure}[htp]
  \centering
  \includegraphics[scale=.6]{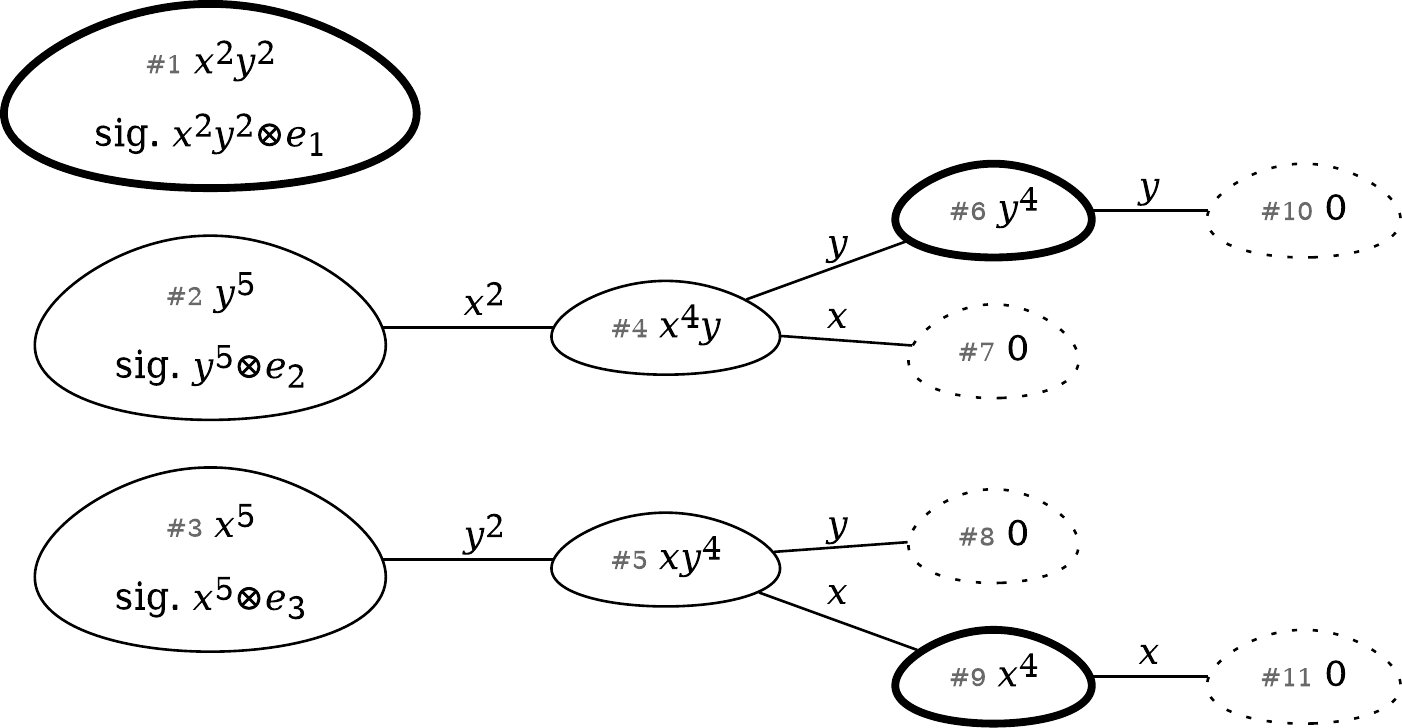}
  \caption{Graphical trace of Algorithm~\ref{algo:rb-min-sig} applied to the
    system in Example~\ref{expl:mora-system}.
    Root nodes, on the left, represent input polynomials.
    Bold nodes represent elements of the \rb{} whose
    leading monomial is the leading monomial of some element of the reduced \gb{} of the input
    ideal. The signature of a node~$n$ can be obtained by multiplying the labels of the edges from~$n$ to the root node, and then multiplying by the signature
    of the corresponding root node.
    For example, the signature of the input node~3 is~$x^5\otimes e_3$, so the signature of the node~11 is~$x^7 y^2 \otimes e_3$.
  }
  \label{fig:mora}
\end{figure}
\begin{figure}[ht]
  \centering
  \includegraphics[scale=.6]{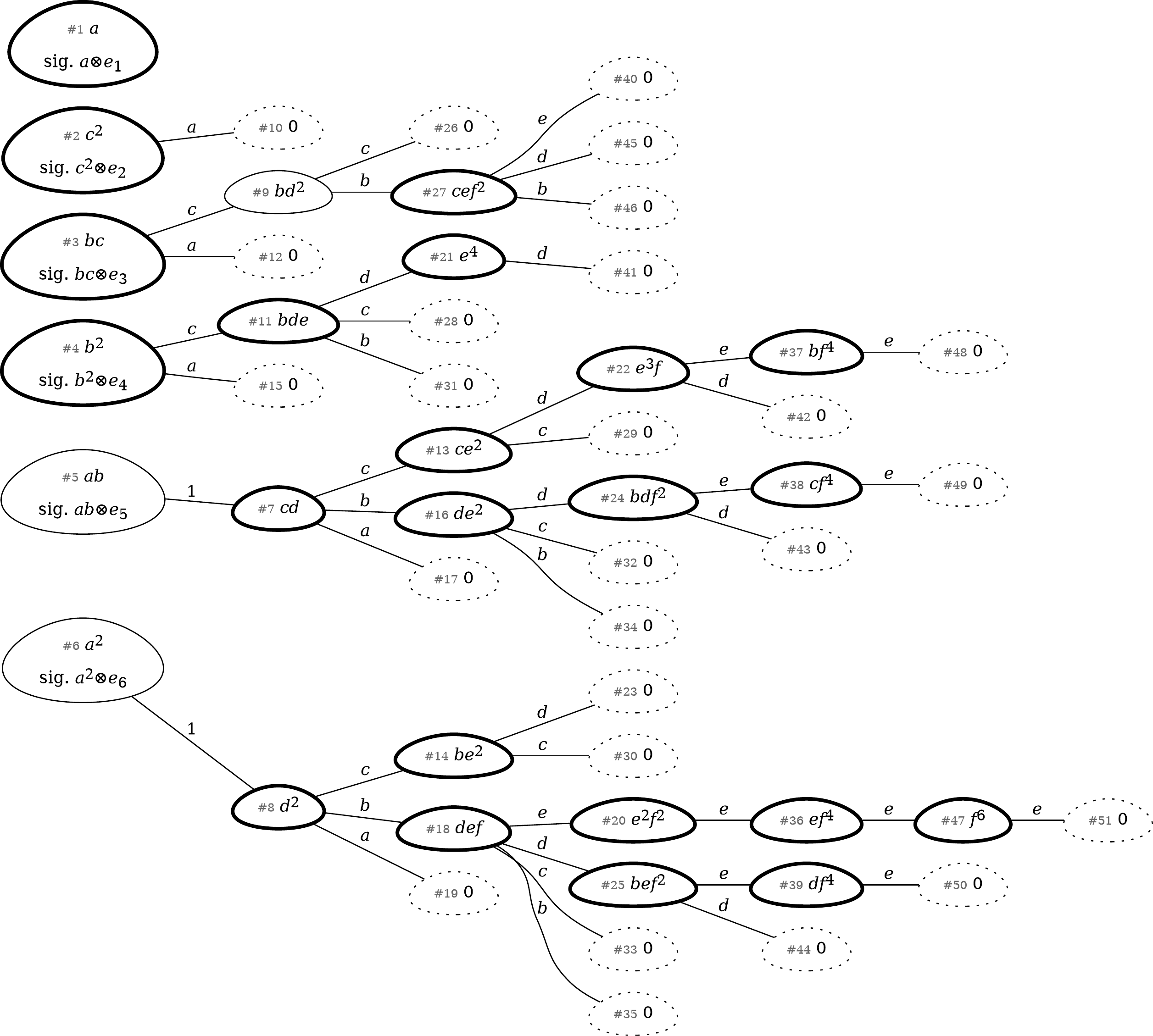}
   \caption{Trace of the computation of a \rb{} for Katsura-6 (Example~\ref{expl:katsura6})
     with the TOP order on the signatures, and the F5 selection strategy of the reductant.
     The input polynomials are given the signatures $\sig g_i = \lm g_i \otimes e_i$.
   }
   \label{fig:top}
\end{figure}

\begin{figure}[htp]
  \centering
  \includegraphics[scale=.6]{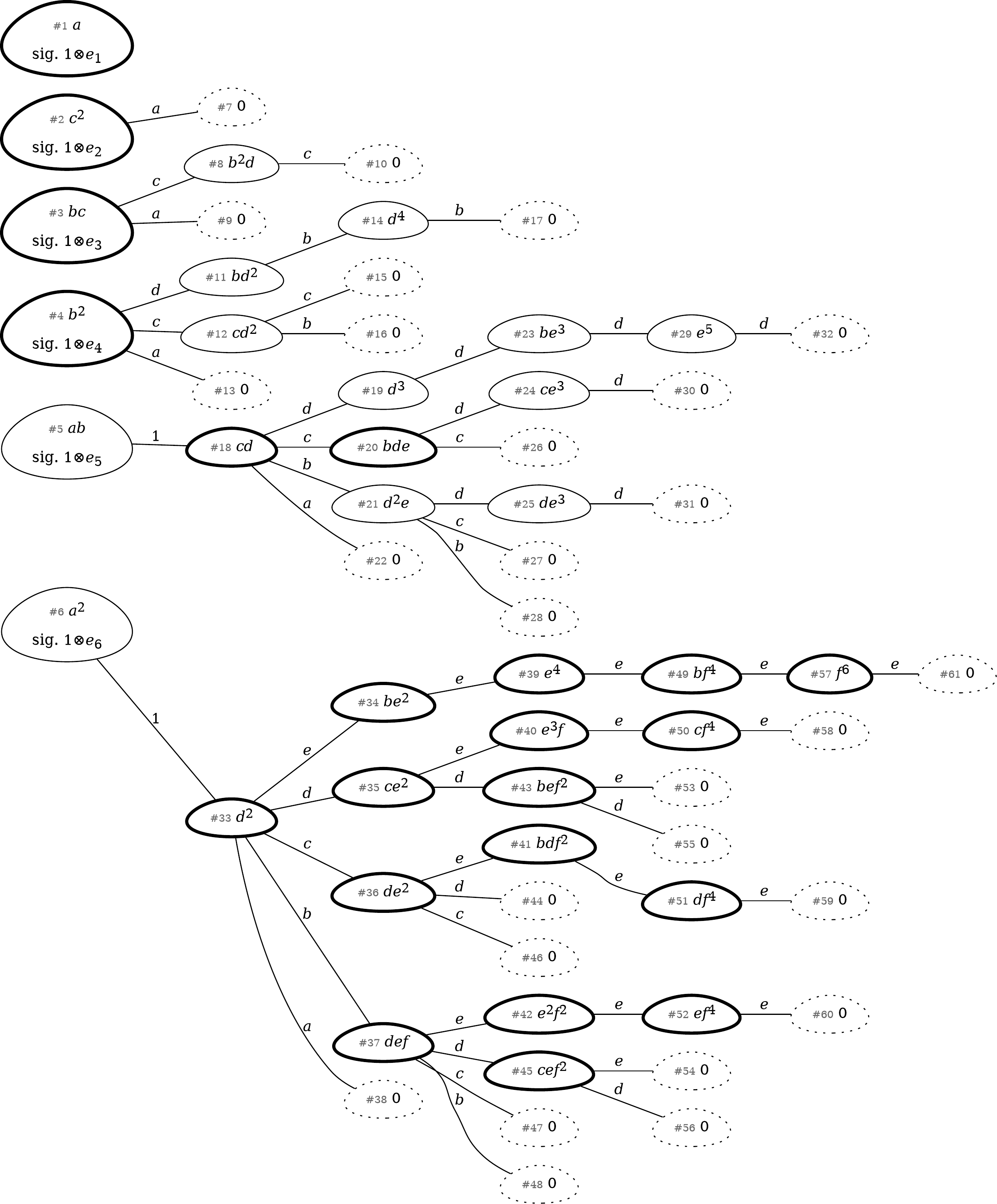}
  \caption{Trace of the computation of a \rb{} for Katsura-6 (Example~\ref{expl:katsura6})
    with the POT order on the signatures, and the F5 selection strategy of the reductant.
    The input polynomials are given the signatures $\sig g_i = 1 \otimes e_i$.
  }
   \label{fig:pot}
\end{figure}

\begin{figure}[htp]
  \centering
  \includegraphics[scale=.45]{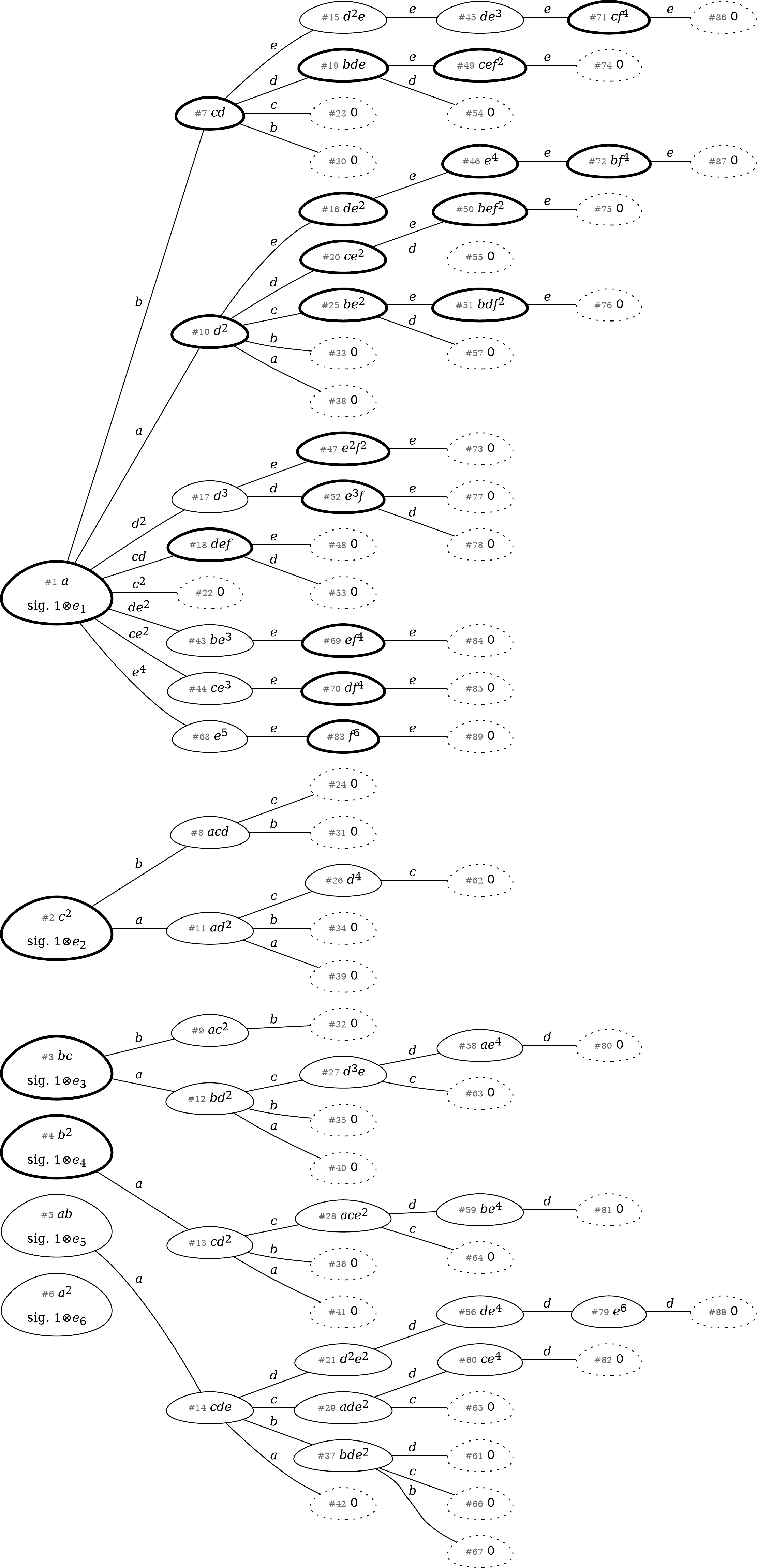}
  \caption{Trace of the computation of a \rb{} for Katsura-6 (Example~\ref{expl:katsura6})
    with the TOP order on the signatures, and the F5 selection strategy of the reductant.
    The input polynomials are given the unshifted signatures $\sig g_i = 1\otimes e_i$.
    In contrast to Figure~\ref{fig:top}, note that many elements are not necessary to form a \gb, eventhough they are necessary to form a \rb.
    For example, the 79th element~$g_{79}$ with~$\lm g_{79} = e^6$ and~$\sig g_{79} = ad^3 \otimes e_4$ cannot be reduced by the earlier~$g_{46}$,
    as suggested by the relation~$e^2\lm g_{46} = \lm g_{79}$, because~$e^2 \sig g_{46} = a e^4 \otimes e_1$ is bigger than~$\sig g_{79}$.}
   \label{fig:unshifted-top}
\end{figure}

\begin{example}[Katsura-6]\label{expl:katsura6}

  We consider the system Katsura-6 (from the famous benchmark family Katsura-$n$, available in Sagemath with the function \texttt{sage.rings.ideal.Katsura}), given in~$\mathbb{Q}[a,b,c,d,e,f]$ (with degree reverse lexicographic ordering) by
  the polynomials
  \begin{align*}
    g_1 &= a + 2 b + 2 c + 2 d + 2 e + 2 f - 1, &
    g_2 &= c^{2} + 2 b d + 2 a e + 2 b f - e,\\
    g_3 &= b c + a d + b e + c f - \tfrac{1}{2} d, &
    g_4 &= b^{2} + 2 a c + 2 b d + 2 c e + 2 d f - c,\\
    g_5 &= a b + b c + c d + d e + e f - \tfrac{1}{2} b, &
    g_6 &= a^{2} + 2 b^{2} + 2 c^{2} + 2 d^{2} + 2 e^{2} + 2 f^{2} - a.
  \end{align*}
  Figures~\ref{fig:top}, \ref{fig:pot} and~\ref{fig:unshifted-top} shows the result of running Algorithm~\ref{algo:rb-min-sig}
  on this input, with different signature orderings. At each iteration, when there are multiple possible choices,
  we pick the one that comes from the most recently inserted element of~$G$, this is the F5 selection strategy, see Section~\ref{sec:f5-reduct-select}.
  The computations are displayed in the form of sigtrees, as in Example~\ref{expl:mora-system-4}.
\end{example}

\subsection{Minimizing the leading monomial of the reductant}
\label{sec:minim-lead-monom}

Processing critical signatures in increasing order seems to be a natural option
but it is also important to understand what happens when signatures are processed in any order.
There may be various reasons to do so: parallel computing, simultaneous reduction in the F4 style \parencite[]{Faugere_1999}.
Recently, \textcite{EderLairezMohrSafeyElDin_2023} used signature algorithms to compute saturation ideals,
this involves enlarging the input ideal on the fly. It can be interpreted as an algorithm processing signatures out of order.

In a time where the termination of F5 \parencite{Galkin_2014} was still unsettled, \Textcite{ArriPerry_2011,ArriPerry_2017}
introduced the idea of choosing carefully the sigpair to be reduced, called the \margindef{reductant}, at a given signature to ensure termination.
This is based on the following observation.

\begin{proposition}
  \label{prop:minimal-lm-implies-dominant}
  Let~$G$ be a prebasis and let~$\sigma$ be a signature at which~$G$ is not a \rb{}.
  Let~$f \in AG$ with~$\sig f = \sigma$ and~$\lm f$ minimal.
  Let~$h$ be a~$\to_G$-normal form of~$f$.
  Then~$g \not\sqsubseteq h$ for any~$g \in G$.
\end{proposition}

\begin{proof}
  By contradiction, assume that there is some~$g\in G$ such that~$g \sqsubseteq h$.
  Condition~\ref{it:dom:non-topred} is ruled out because~$h$ is $\to_G$-reduced.
  Therefore~\ref{it:dom:non-min} holds: there is some~$a\in A$ such that~$a \sig g = \sigma$ and~$\lm(ag) \leq \lm h$.
  Besides,  $G$ is not a \rb{} at~$\sigma$, it follows that~$f$ is not~$\to_G$-reduced,
  and thus~$\lm h < \lm f$ since~$f \to_G h$.
  It follows that~$\lm(ag) < \lm f$, which contradicts the minimality of~$f$.
\end{proof}

Although \citeauthor[]{ArriPerry_2011}
still requires to process signatures by increasing order,
Proposition~\ref{prop:minimal-lm-implies-dominant} opens the way for out-of-order signature handling,
as \textcite{SunWang_2013} and~\textcite{GaoVolnyWang_2016} did.
The formulation that proposed here (Algorithm~\ref{algo:rb-min-lm})
is mostly equivalent to that of the latter.
The choice of signature on line~\ref{line:minlm:picksig} is unconstrained, but the choice of the reductant is imposed.

\begin{algorithm}
  \begin{description}
    \item[input] A finite prebasis~$G$
    \item[output] A finite sigsafe extension of~$G$ which is a \rb{}
  \end{description}
  \begin{pseudo}
    \kw{while} $G$ is not a \rb{} at all $\sigma\in \Sigma(G)$ \kw{do}\\+
    pick any~$\sigma\in \SIG$ such that~$G$ is not a \rb{} at~$\sigma$ \label{line:minlm:picksig}\\
    pick~$f\in AG$ with~$\sig f = \sigma$ and~$\lm f$ minimal\\
    $g \gets $ any $\to_G$-normal form of~$f$ \\
    $G \gets G \cup \left\{ g \right\}$\\-
    \kw{return} $G$
  \end{pseudo}
  \caption{Computation of a \rb{} with out-of-order signature processing, minimizing the leading monomial of the reductant}
  \label{algo:rb-min-lm}
\end{algorithm}

\begin{theorem}\label{thm:algo-minlm}
  Algorithm~\ref{algo:rb-min-lm} is correct and terminates.
\end{theorem}

\begin{proof}
  Identical to the proof of Theorem~\ref{thm:algo-minsig}, but using Proposition~\ref{prop:minimal-lm-implies-dominant}
  instead of Proposition~\ref{prop:minimal-sig-implies-dominant}.
\end{proof}

\subsection{Well-formed sigtrees}
\label{sec:well-formed-sigtrees}

A \emph{tree} is a set~$\mathcal{T}$, finite or infinite, of finite sequences of nonnegative integers
such that for any $(k_1,\dotsc,k_r) \in \mathcal{T}$ (with~$r \geq 1$),
the prefix subsequence~$(k_1,\dotsc,k_{r-1})$ is also in~$\mathcal{T}$.
The elements of~$\mathcal{T}$ are called \emph{nodes}.
The \emph{children} of a given node~$n\in \mathcal{T}$ are the sequences in~$\mathcal{T}$ that extend~$n$ by exactly one integer.
The \emph{ancestors} of a node~$n = (k_1,\dotsc,k_r)$ are the nodes~$(k_1,\dotsc,k_j)$ for~$0 \leq j < r$.

\begin{definition}[sigtree]
  A \emph{sigtree} is a tree $\mathcal{T}$ together with a rank function~$\rk : \mathcal{T} \to \mathbb{N}$ and a label funtion~$\lambda:\mathcal{T} \to M\times \SIG$
  (recall that~$M\times \SIG$ is the set of sigpairs).
\end{definition}
Sigtrees are a natural way to represent the process of computing a rewrite basis.
Indeed, Algorithms~\ref{algo:rb-min-lm} or~\ref{algo:rb-min-sig}, as well as Pseudo-algorithm~\ref{algo:abstract-algo},
produce sigtrees as follows.
The elements of~$G$, the \rb{} begin computed, are the labels of the sigtrees.
There is one sigtree for each element of the input sigset.
The root node of each sigtree is labelled with the corresponding input element.
At the beginning of the algorithms, there are only the root nodes.
Then, each time some~$f \in AG$ is picked,
reduced, and inserted into~$G$, we can write~$f = a \lambda(n)$, for some node~$n$ of the sigtree, and some~$a\in A$,
and we insert in the sigtree a new child node of~$n$ containing the new element.
The rank function reflects the \emph{birthdate} of a node.
Figures~\ref{fig:mora}, \ref{fig:top}, \ref{fig:pot} and~\ref{fig:unshifted-top} are examples of sigtrees obtained in this way.

\begin{definition}[well-formed sigtree]
  A \emph{well-formed sigtree}
  is a sigtree~$\mathcal{T}$ such that:
  \begin{enumerate}[label=T{\arabic*}]
    \item\label{it:wfst:ratio} $\forall n \in \mathcal{T}, \forall m \text{ child of $n$}, \exists a \in A, a \sig \lambda (n) = \sig \lambda(m)$ and $a \lm \lambda(n) > \lm\lambda(m)$,\\
          \emph{``a child is more reduced than its parent''};
    \item \label{it:wfst:ancestors}
          $\forall n \in \mathcal{T}$, $\lambda(n)$ is $\to$-reduced with respect to the sigset~$ \left\{ \lambda(p) \st \text{$p$ is an ancestor of~$n$} \right\}$,\\
          \emph{``a child is reduced modulo its ancestors''};
    \item \label{it:wfst:siblings} $\forall n \in \mathcal{T},\forall p, q $ children of~$n$, $\rk(p) < \rk(q) \Rightarrow \sig \lambda(p)$ does not divide $\sig\lambda(q)$,\\
          \emph{``the signature of a node does not divide that of younger sibling nodes''}.

    \item \label{it:wfst:rklevels} $\forall k\in \mathbb{N}, \left\{ n\in \mathcal{T} \st \rk(n) = k \right\}$ is finite.
  \end{enumerate}
\end{definition}
Typically, \ref{it:wfst:ratio} and~\ref{it:wfst:ancestors} are satisfied \emph{by design} if $\lambda(m)$ is obtained by reducing~$a \lambda(n)$ modulo a sigset containing at least the labels of the ancestors of~$m$,
and assuming that~$a \lambda(n)$ is indeed reducible to account for the strict inequality in~\ref{it:wfst:ratio}.
\ref{it:wfst:rklevels} will follow from an appropriate bookkeeping.
\ref{it:wfst:siblings} is the real constraint. In the context above, \ref{it:wfst:siblings} puts a contraint on the choice of the reductant. It means that whenever we want to reduce~$a \lambda(n)$, we must first check that we have not previously computed a reduction~$b \lambda(n) \to \lambda(m)$ for some child $m$ of~$n$ and for some~$b$ such that~$\exists c\in A, cb \lambda(n) = a \lambda(n)$. In which case we can reduce~$c \lambda(m)$ instead of~$a \lambda(n)$.

\begin{theorem}\label{thm:sigtree-term}
  A well-formed sigtree is finite.
\end{theorem}

\begin{proof}
  Let~$\mathcal{T}$ be a well-formed sigtree.
  By König's lemma, it is enough to prove that~$\mathcal{T}$ has no infinite branch
  and that every node has at most finitely many children.

  If there is an infinite branch, then there is a sequence of nodes~$(n_i)_{i \geq 0}$
  such that~$n_{i+1}$ is a child of~$n_i$.
  By Proposition~\ref{prop:dom-noetherian}, there are some indices~$i < j$ such that~$\lambda(n_i) \sqsubseteq \lambda(n_j)$.
  Condition~\ref{it:dom:non-topred} would contradict~\ref{it:wfst:ancestors}
  so Condition~\ref{it:dom:non-min} holds: there is some~$b\in A$ such that~$b\sig \lambda(n_i) = \sig \lambda(n_j)$
  and~$b \lm \lambda(n_i) \leq \lm \lambda(n_j)$.
  By \ref{it:wfst:ratio} (applied all along the path from~$n_i$ to~$n_j$),
  there is some~$a \in A$ such that~$a\sig \lambda(n_i) = \sig\lambda(n_j)$ and~$a \lm\lambda(n_i) > \lm \lambda(n_j)$.
  Since~$a\sig \lambda(n_i) = b\sig \lambda(n_i)$, \ref{it:sig:refines} implies that~$a\lm \lambda(n_i) = b \lm \lambda(n_i)$, leading to a contradiction.

  If a node has infinitely many children,
  \ref{it:wfst:rklevels} ensures that we can extract an infinite sequence of children with increasing ranks.
  By Noetherianity of~$S$, the signature of one child would divide the signature of another with higher rank.
  This contradicts~\ref{it:wfst:siblings}.
\end{proof}

\subsection{The F5 reductant selection strategy}
\label{sec:f5-reduct-select}

In the original presentation of F5, \textcite{Faugere_2002}
proposes to choose a reductant~$a f$ where, among all possible choices, $f$ is the ``most recent''.
This leads to Algorithm~\ref{algo:rb-f5}.
This selection strategy leads naturally to a well-formed sigtree.
So we can prove that the corresponding algorithm terminates, even if signatures are handled out of order.

\begin{algorithm}[ht]
  \begin{description}
    \item[input] A finite prebasis~$G$
    \item[output] A finite sigsafe extension of~$G$ which is a \rb{}
  \end{description}
  \begin{pseudo}
    $R\gets $ empty dictionary mapping sigpairs to integers\\
    $r \gets 1$ \\
    \kw{for} $g\in G$ \kw{do} $R[g] \gets 0$\\
    \kw{while} $G$ is not a \rb{} at all $\sigma\in \Sigma(G)$ \kw{do}\\+
    pick any~$\sigma\in \SIG$ such that~$G$ is not a \rb{} at~$\sigma$ \\
    pick some~$a \in A$ and~$f\in G$ such that~$a \sig f = \sigma$ and~$R[f]$ maximal \label{line:f5:reductant}\\
    $g \gets $ any $\to_G$-normal form of~$af$ \\
    $G \gets G \cup \left\{ g \right\}$ \label{line:f5:insert}\\
    $R[g] \gets r$\\
    $r\gets r+1$\\-
    \kw{return} $G$
  \end{pseudo}
  \caption{Computation of a \rb{} with out-of-order signature processing and F5 selection strategy of the reductant}
  \label{algo:rb-f5}
\end{algorithm}

\begin{theorem}
  Algorithm~\ref{algo:rb-f5} is correct and terminates.
\end{theorem}

\begin{proof}
  Correctness follows from Theorem~\ref{thm:faugere-criterion}.
  For termination, consider the sigtrees (one for each input element) induced by the algorithm:
  each sigpair $g$ inserted into~$G$ on line~\ref{line:f5:insert}
  is the label of a node whose parent is the node labeled with~$f$,
  where~$f$ is the sigpair picked on line~\ref{line:f5:reductant}.
  The rank of a node is given by~$R$.
  If the algorithm does not terminate, at least one of the sigtrees is infinite.
  Therefore, to prove termination, it is enough to check that the sigtrees are finite.

  These sigtrees are well-formed. \ref{it:wfst:ratio} and~\ref{it:wfst:ancestors} follow by construction.
  To check~\ref{it:wfst:siblings}, we observe that the rank of a node is always greater than the rank of its parent.
  So, on line~\ref{line:f5:reductant},
  if the node corresponding to~$f$ has already a child whose signature divides~$\sigma$, this child has a higher rank than that of~$f$,
  which contradicts the maximality of~$R[f]$.
  The number of nodes of a given rank is at most one, this gives~\ref{it:wfst:rklevels}. Theorem~\ref{thm:sigtree-term} applies and shows that the sigtrees are finite,
  so the algorithm terminates.
\end{proof}

\subsection{Explicit management of the critical set}
\label{sec:expl-manag-crit}

The presentation of Algorithms~\ref{algo:rb-min-sig}, \ref{algo:rb-min-lm} and~\ref{algo:rb-f5}
takes advantage of the notation~$\Sigma(G)$ to abstract the handling of set of signatures to be handled
from concrete questions that theory may ignore but not practical implementations.
There is a lot of room to design a proper handling of signatures, I simply show some possible variants.

\begin{algorithm}[p]
  \begin{pseudo}
    \inschunk{init}\\
    \kw{while} $Q$ is not empty \kw{do} \\+
    $\sigma \gets $ some element of~$Q$\\
    $Q \gets Q \setminus \left\{ \sigma \right\}$\\
    \inschunk{selectreductant} \\
    \kw{if} $f$ is $\to_G$-reducible \kw{then}\\+
    $g \gets $ a $\to_G$-normal form of~$f$\\
    \inschunk{insertnode}\\
    \inschunk{updatequeue}\\
    $G \gets G\cup \left\{ g \right\}$\\--
    \kw{return} G\\
    \\
    \ctfont{Chunks}\\
    \defchunk{init}{initialize signature queue and sigtree} \\+
    $Q \gets \varnothing$ \ct{signature queue}\\
    $\children \gets $ empty list \ct{maps a node to its children}\\
    $\labell \gets $ empty list \ct{maps a node to its label}\\
    $\children[0] \gets \varnothing$ \ct{the set of root nodes}\\
    $n \gets 1$    \ct{node counter} \\
    $k\gets 0$ \ct{index of the root node}\\
    \kw{for} $g \in G$ \kw{do} \ct{create nodes for input elements}\\+
    \inschunk{insertnode}\\
    \inschunk{updatequeue}\\--

    \\
    \defchunk{insertnode}{insert a node with label $g$ and parent~$k$}\\+
    $\labell[n] \gets g$ ; \\
    $\children[k] \gets \children[k] \cup \{ n\}$ ; \\
    $\children[n] \gets \varnothing$\\
    $n \gets n+1$\\-

    \\
    \defchunk{updatequeue}{update the queue with the new relation~$g$}\\+
    \kw{for} $h\in G$ \kw{do} \\+
    $Q \gets Q \cup \Sigma(g, h) \cup \Sigma(h, g)$\\--

    \\
    \defchunk{selectreductant}{select a reductant~$f$ in signature~$\sigma$ with corresponding node~$k$} \\+
    $k \gets 0$             \ct{start the search from the root node}\\
    \kw{for} $c \in \children[k]$ \kw{do} \label{line:select:reductant} \ct{the order of iteration does not matter}\\+
    \kw{if} $\sig L[c]$ divides $\sigma$ \kw{then}\\+
    $k \gets c$ \ct{go down the tree}\\
    \kw{goto} \ref{line:select:reductant} \ct{continue the search from the new position}\\--
    pick~$a\in A$ such that~$a \sig \labell[k] = \sigma$\\
    $f \gets a \labell[k]$
  \end{pseudo}
  \caption{Computation of a \rb{}, with explicit construction of a well-formed sigtree and explicit handling of the critical set}
  \label{algo:explicit}
\end{algorithm}

\subsubsection{Base algorithm}
\label{sec:base-algorithm}

In this section, we assume that we know how to operate on~$\basis$ and~$\SIG$ (that is compare, test divisibility, etc.) and we assume that we have a procedure to compute the critical set~$\Sigma(f, \left\{ g \right\})$ of a pair of sigpairs~$f$ and~$g$ (simply denoted~$\Sigma(f, g)$).
Without more information on~$A$, $\basis$ and~$\SIG$ we cannot go further down into the details.
In the polynomial setting, the set~$\Sigma(f, g)$ may contain zero or one element and its computation amounts to a few operations on monomials, see Section~\ref{sec:polynomial-ring},

There are many ways to proceed and Algorithm~\ref{algo:explicit} is one of them.
In this algorithm, the set~$Q$ contains signatures, and, at the beginning of each iteration of the ``while'' loop,
we have the following invariant:
\begin{equation}
  \label{eq:2}
  \forall \sigma \in \Sigma(G), \sigma \in Q \text{ or $G$ is a \rb{} at~$\sigma$}.
\end{equation}
Indeed, when an element~$g$ is inserted in~$G$, we remove~$\sig g$ from~$Q$
and insert all the elements in the sets~$\Sigma(g, h)\cup \Sigma(h, g)$, for~$h\in G$.
Since~$g$ is~$\to_G$ reduced, 
$G\cup \left\{ g \right\}$ is a \rb{} at~$\sigma$
and the inclusion
\[ \Sigma(G \cup \left\{ g \right\}) \subseteq \Sigma(G) \cup \bigcup_{h\in G} \left( \Sigma(g, h) \cup \Sigma(h, g) \right) \]
proves that Invariant~\eqref{eq:2} is preserved.
With Invariant~\eqref{eq:2} and Theorem~\ref{thm:faugere-criterion} in hand,
it is clear that Algorithm~\ref{algo:explicit} returns a \rb{} when it terminates.

Termination is ensured \emph{by design} by contructing well-formed sigtrees.
The algorithm maintains two lists~$\children$ and $L$.
The~$L$ list contains the labels: $L[i]$ is the label of the~$i$th node in the sigtree.
The~$\children$ list encodes the tree structure: $\children[i]$ is the set of chidren of the node~$i$.
The set~$\children[0]$ contains the root nodes.
The rank of the $i$th node is defined to be~$i$.
The selection procedure of the reductant makes it sure that the sigtree is well formed.
Each iteration of the ``while'' loop either removes an element of~$Q$ or increase the size of the sigtree.
The latter cannot happen infinitely many times, in view of Theorem~\ref{thm:sigtree-term},
so~$Q$ is eventually empty and the algorithm terminates.

\subsubsection{F5 variant}

We can specialize the reductant selection strategy to match the one of F5, exposed in Section~\ref{sec:f5-reduct-select}.
In this variant, it is not necessary to maintain the sigtree explicitely, we may ignore the $\children$ list.
(To really match with F5 algorithm, the reductant is chosen to be zero if possible, even if it does not correspond to the most recent possible reductant.)

\begin{algorithm}[htp]
  \begin{pseudo}
    \ctfont{Similar to Algorithm~\ref{algo:explicit}, except for the following chunk}\\
    \defchunk{selectreductantf5-variant}{select a reductant~$f$ in signature~$\sigma$ with corresponding node~$k$}\\+
    $k \gets 0$\\
    \kw{for} $1\leq j < n$ \kw{do} \\+
    \kw{if} $\sig L[j]$ divides~$\sigma$ \kw{then}
    $k \gets j$\\
    \kw{if} $\lm L[j] = 0$ \kw{then} \kw{break}  \ct{stop the search if $\sigma$ is a syzygy signature}\\-
    pick~$a\in A$ such that~$a \sig L[k] = \sigma$\\
    $f \gets a L[k]$
  \end{pseudo}
  \caption{Variant of Algorithm~\ref{algo:explicit} with the F5 strategy for the reductant selection}
\end{algorithm}

\subsubsection{A variant with signature pruning}

In the set~$Q$, we may remove any element that is divided by a different element of~$Q$.
Instead of Invariant~\eqref{eq:2}, we maintain the following one:
\begin{equation}
  \forall \sigma\in \Sigma(G), \left( \exists \tau\in Q, \tau \text{ divides } \sigma \right) \text{ or $G$ is a \rb{} at $\sigma$}.
\end{equation}
This leads to Algorithm~\ref{algo:explicit:pruning}.
Checking correctness is an easy exercise.

\begin{algorithm}[ht]
  \begin{pseudo}
    \ctfont{Similar to Algorithm~\ref{algo:explicit}, except for the following chunk}\\
    \defchunk{updatequeue-variant}{update the queue with the new relation~$g$}\\+
    \kw{for} $h\in G$ \kw{do} \\+
    $Q \gets Q \cup \Sigma(g, h) \cup \Sigma(h, g)$\\-
    \kw{for} $\sigma\in Q$ \kw{do}\\+
    \kw{if} $\exists \tau \in Q\setminus \left\{ \sigma \right\}$, $\tau$ divides~$\sigma$ \kw{then}\\+
    $Q\gets Q\setminus \left\{ \sigma \right\}$
  \end{pseudo}
  \caption{Variant of Algorithm~\ref{algo:explicit} with signature pruning}
  \label{algo:explicit:pruning}
\end{algorithm}

\subsection{Simultaneous reduction}
\label{sec:simult-reduct}

\begin{algorithm}[ht]
  \begin{pseudo}
    \inschunk{init} \\
    \kw{while} $Q$ is not empty \kw{do} \\+
    $S \gets $ some nonempty subset of~$Q$ \ct{select several signatures at a time}\\
    $Q\gets Q \setminus S$\\
    $F \gets \varnothing$ \ct{set of reductants and corresponding nodes} \\
    \kw{for} $\sigma\in S$ \kw{do} \ct{selection of reductants} \\+
    \inschunk{selectreductant} \\
    \kw{if} $f$ is~$\to_G$-reducible \kw{then}\\+
    $F \gets F \cup \left\{ (f, k) \right\}$ \ct{we keep the information of the parent}\\--

    $N \gets \varnothing$ \ct{set of newly computed relations}\\
    \kw{for} $(f, k)\in F$ by increasing order of~$\sig f$ \kw{do} \ct{reduction of reductants}\\+
    $g^\natural \gets $ a $\to$-normal form of~$f^\natural$ w.r.t. $\AGlt \cup \left\{ \smash{h^\natural} \st h\in N \right\}$ \label{line:f5f4:reduction}\\
    $g \gets (g^\natural, \sigma)$\\
    \inschunk{insertnode}\\
    \inschunk{updatequeue}\\
    $N\gets N \cup \left\{ g \right\}$ \ct{insertion of~$g$ in~$G$ is delayed} \\-

    $G \gets G\cup N$\\
    \kw{return} G
  \end{pseudo}
  \caption{Computation of a \rb{}, with simultaneous reduction. Pseudocode chunks are defined in Algorithm~\ref{algo:explicit}.}
  \label{algo:f5f4}
\end{algorithm}

As another variation of Algorithm~\ref{algo:explicit},
we may handle several signatures at a time,
in the F4 style~\parencites{Faugere_1999,AlbrechtPerry_2010}[\S13]{EderFaugere_2017}.
Concretely, the sigset~$G$ that is used to compute the reductions not
updated each time a new element is discovered.
The new elements are inserted in a sigset~$N$
and after a bunch of signatures is handled (how many is to be determined by the implementation), the elements of~$N$ are inserted in~$G$.
On line~\ref{line:f5f4:reduction}, the reductant~$g$ is reduced with respect to~$G$ (and as usual, multiples of elements of~$G$ can be used in reduction steps) and also with respect to~$N$ (but multiples cannot be used in reduction steps). In other words, the polynomial part~$f^\natural$ is reduced modulo the set~$\AGlt \cup N$.

The reason to delay insertion into~$G$ is the principle of simultenous reduction.
If we have to perform the reductions of sigpairs~$f_1,\dotsc,f_r$ with respect to the same sigset~$G$,
it is possible to formulate the problem in terms of a matrix whose rows represent the~$f_i$ and all possibly useful reducers from~$AG$
in a reduction chain starting from any of the~$f_i$.
Once this matrix is computed (this is the \emph{symbolic preprocessing} step), it can be used to compute the reductions efficiently.
This matrix aspect is crucial for high-performance computations but it is a transparent transformation of the algorithm: it does not change what is computed, compared to the naive reductions of the~$f_i$.
For a more detailed introduction to the F4 strategy, see \textcite[Chapter 10, \S 3]{CoxLittleOShea_2015}.

\begin{theorem}
  Algorithm~\ref{algo:f5f4} is correct and terminates.
\end{theorem}

\begin{proof}
  Termination is clear because the algorithm produces a well-formed sigtree (where the rank of a node is the number of the iteration at which it was produced),
  and at each iteration, either $Q$ diminishes or the sigtree grows.
  Correctness follows from Invariant~\eqref{eq:2} which also holds for this algorithm,
  with a slightly different argument than the one in Section~\ref{sec:base-algorithm}. Indeed,
  when an element~$g$ is inserted into~$G$,
  if~$g$ is~$\to_G$-reduced, then
  $G\cup \left\{ g \right\}$ is a \rb{} at~$\sigma$
  so we may remove~$\sigma$ from~$Q$ without breaking the invariant.
  However, due to the nature of simultaneous reduction, it may happen that we insert an element that is not~$\to_G$-reduced.
  In this case, then there is some~$h\in G$ which reduces~$g$ and we check easily that~$\Sigma(g, h) = \left\{ \sig g \right\}$.
  So in this case, $\sig g$ is not actually removed from~$Q$ and the invariant is preserved.
\end{proof}

\section{Settings}
\label{sec:settings}

This section describes different monomial spaces coming from different settings in computer algebra. Some are noncommutative or non-Noetherian.

\subsection{Polynomial ring}
\label{sec:polynomial-ring}

Let~$M = K[x_1,\dotsc,x_n]$ be the polynomial ring in~$n$ variables over~$K$,
which we endow with a monomial order, so the function~$\lm$ is well defined.
Let~$A = \left\{ \smash{x_1^{i_1}\dotsb x_n^{i_n}} \st i_1,\dotsc,i_n \in \mathbb{N} \right\}$.
The axioms for monomial orders ensure that~$M$ is a monomial module over~$A$.
It is Noetherian.
Moreover, it satisfies the extra property~\ref{it:momo:extra},
so construction of prebases is easy, see Remark~\ref{remark:prebases-for-free}.

For sigpairs~$f$ and~$g$, the critical set~$\Sigma(f, g)$ has zero or one element.
There is the trivial case where~$f^\natural = 0$ or~$g^\natural = 0$. In this case, every multiple of~$f$ is $\to_{ \left\{ g \right\}}$-reduced, so~$\Sigma(f, g) = \varnothing$.
When~$f^\natural$ and~$g^\natural$ are both nonzero,
there are monomials~$a, b\in A$ such that~$a\lm f = b \lm g = \mop{lcm}(\lm f, \lm g)$.
Then there are two cases, if~$a \sig f \leq b \sig g$, then~$\Sigma(f, \left\{ g \right\}) = \varnothing$;
on the contrary, if~$b \sig g < a \sig f$, then~$\Sigma(f, g) = \left\{ a\sig f \right\}$.

\subsection{Modules over polynomial rings}

Let~$r$ be a positive integer and let~$M = K[x_1,\dotsc,x_n]^r$,
which we endow with a term ordering --~typically position-over-term, term-over-position, or Schreyer's order \parencite[\S1.4]{KreuzerRobbiano_2000}.
The monoid~$A$ is the same as before.
$M$ is a Noetherian monomial module and satisfies the extra condition~\ref{it:momo:extra}.

The computation of~$\Sigma(f, g)$ is slightly different.
In the case where~$f^\natural$ and~$g^\natural$ are both nonzero, it may happen that no multiple of~$\lm f$ and~$\lm g$ coincide. Indeed, nonzero monomials in~$\basis$ have an \emph{index} in~$ \left\{ 1,\dotsc,r \right\}$
which is unchanged under multiplication.
Therefore, if~$f^\natural = 0$ or~$g^\natural = 0$, or~$\lm f$ and~$\lm g$ have different indices, then~$\Sigma(f, g) = \varnothing$.
Otherwise, there are monomial~$a, b\in A$ such that~$a\lm f = b\lm g$ (and~$a\lm f$ minimal).
Depending on the comparison of~$a\sig f$ and~$b\sig g$, $\Sigma(f, g)$ is either~$\varnothing$ or~$\left\{ a\sig f \right\}$, as in the polynomial case.

\subsection{Monoid algebras}

Let~$A$ be a submonoid of
$\left\{ \smash{x_1^{i_1}\dotsb x_n^{i_n}} \st i_1,\dotsc,i_n \in \mathbb{N} \right\}$
and let~$M = K[A]$ be the ring of polynomials whose monomials are contained in~$A$.
It is clear that~$M$ is a Noetherian monomial module over~$A$.
This case includes the ``semigroup algebras'' studied by \textcite{BenderFaugereTsigaridas_2019}.
It also includes some algebras that are interesting in singularity theory such that~$K[x^2, xy, y^2]$,
that are polynomial ring with finitelty many monomials removed (in this case~$x$, $y$, and $xy$).


The critical set~$\Sigma(f, g)$ can contain more than one element.
Assume, for example, that~$M = K[x^2, xy, y^2]$ --~that is~$A = \left\{ x^i y^j \st i+j \geq 2 \right\}$~--
and that~$f^\natural = x^2$ and~$g^\natural = xy$.
The set of all~$a\in A$ such that~$\lm(af^\natural)$ is divided by~$\lm(g^\natural)$
is generated by~$xy$ and~$y^2$. It is not generated by~$y$ because~$y$ is not in~$A$.
Assuming that~$xy\sig f > x^2 \sig g$ and~$y^2 \sig f > xy \sig g$, we have
\[ \Sigma(f, g) = \left\{ xy \sig f, y^2\sig f \right\}. \]

\subsection{Weyl algebras}

Let~$M = K\langle x_1,\dotsc, x_n, \partial_1,\dotsc,\partial_n \rangle$ be the Weyl algebra on $n$ variables.
It is noncommutative.
We may define it as the subalgebra of $\mop{End}_K( K[X_1,\dotsc,X_n] )$
where~$x_i$ is the multiplication by~$X_i$ and~$\partial_i$ is the differentiation with respect to~$X_i$.
Concretely, $x_i x_j = x_j x_i$, $\partial_i \partial_j = \partial_j \partial_i$,
$\partial_i x_j = x_j \partial_i$ (if~$i\neq j$)
and~$\partial_i x_i = x_i \partial_i + 1$.
A basis of~$M$ is given by the monomials~$x_1^{i_1}\dotsb x_n^{i_n} \partial_1^{j_1} \dotsb \partial_n^{j_n}$
and we can consider the same monomial orderings as we would do for a commutative polynomial ring in~$2n$ variables.

For the monoid~$A$, we cannot choose the set of monomials because it is not closed under multiplication.
We choose instead~$A$ to be the submonoid of~$M$ generated by~$x_1,\dotsc,x_n$ and~$\partial_1,\dotsc,\partial_n$.
We could also choose~$A = M$.
This turns~$M$ into a Noetherian monomial module with the extra property~\ref{it:momo:extra},
so we can construct signature modules with Remark~\ref{remark:prebases-for-free}.
We could also choose~$A$ to be the monoid of nonzero elements of~$R$.
Things behave similarly to the polynomial case, due to \emph{quasicommutativity}: for any~$a, b \in M$, $\lm(ab) = \lm(ba)$.

\subsection{Differential algebras}
\label{sec:diff-algebr}

Let~$M = K[t, x_0, x_1, x_2, \dotsc]$ be
a polynomial ring in infinitely many variables
with a derivation defined by~$t' = 1$ and $x_i' = x_{i+1}$.
Let~$W = M\langle \partial \rangle$
be the subalgebra of~$\mop{End}_K M$ where~$M$ acts by multiplication
and~$\partial$ be the derivation, similarly to the Weyl algebra case.
This turns~$M$ into a left~$W$-module
and \emph{differential ideals} are defined to be the submodules of~$M$.
We choose on~$M$ a lexicographic ordering with~$t < x_0 < x_1$...

We choose $A$ to be the monoid generated by~$\partial, t, x_0, x_1$...
This turns~$M$ into a monomial module. It is quasicommutative (that is~$\lm(abm) = \lm(bam)$ for any~$a,b\in W$ and~$m\in M$) but not Noetherian.
However, it satisfies the extra condition~\ref{it:momo:extra}
and the critical sets~$\Sigma(f, G)$ are finite.
This example extends to several independent variables and several function variables.

\subsection{Free algebras}
\label{sec:free-algebras}

Let~$M$ be the free algebra generated by~$n$ variables~$x_1,\dotsc,x_n$.
A basis of~$M$ is given by the monoid of words in~$x_1,\dotsc,x_n$.
A monomial order may be given, for example, by comparing the degree first and then the lexicographic order.
We choose~$A$ to be the monoid of words, which acts naturally on~$M$ by left multiplication.
This turns~$M$ into a monomial space with extra condition~\ref{it:momo:extra}.
If~$n > 1$, it is not Noetherian, but the critical sets are finite.

To deal with two-sided ideals of~$M$,
we need to consider not only left multiplications but also right multiplications.
We introduce the monoid~$A' = A \times A^\text{op}$ of pairs of words with
with the composition $(a, b) (a', b') = (aa', b'b)$
and the action on~$M$ given by~$(a, b)m = amb$.
This turns~$M$ into another mononomial space.
When~$n > 1$, it is not Noetherian and does not satisfy extra condition~\ref{it:momo:extra}.
For example, as shown by~\textcite{GreenMoraUfnarovski_1998},
if~$x_1 > x_2$ then~$x_1 x_1 - x_1 x_2$ generates a two-sided ideal without a finite \gb.
Moreoever, the critical sets may be infinite, eventhough they
contain only finitely many nonsyzygy signatures \parencite{HofstadlerVerron_2022}.

\raggedright
\printbibliography

@ONLINE{AlbrechtPerry_2010,
  AUTHOR = {Albrecht, Martin and Perry, John},
  DATE = {2010-10},
  EPRINT = {1006.4933},
  EPRINTTYPE = {arxiv},
  TITLE = {F4/5},
}

@ARTICLE{ArriPerry_2011,
  AUTHOR = {Arri, Alberto and Perry, John},
  DATE = {2011-09-01},
  DOI = {10/cd5td7},
  ISSN = {0747-7171},
  JOURNALTITLE = {J. Symb. Comput.},
  LANGID = {english},
  NUMBER = {9},
  PAGES = {1017--1029},
  TITLE = {The {{F5}} Criterion Revised},
  VOLUME = {46},
}

@ARTICLE{ArriPerry_2017,
  AUTHOR = {Arri, Alberto and Perry, John},
  DATE = {2017-09-01},
  DOI = {10/gp8639},
  ISSN = {0747-7171},
  JOURNALTITLE = {J. Symb. Comput.},
  LANGID = {english},
  PAGES = {164--165},
  TITLE = {Corrigendum to “{{The F5}} Criterion Revised”},
  VOLUME = {82},
}

@ARTICLE{BardetFaugereSalvy_2015,
  AUTHOR = {Bardet, Magali and Faugère, Jean-Charles and Salvy, Bruno},
  DATE = {2015},
  DOI = {10/gntfcb},
  ISSN = {0747-7171},
  JOURNALTITLE = {J. Symb. Comput.},
  PAGES = {49--70},
  TITLE = {On the Complexity of the {{F5 Gröbner}} Basis Algorithm},
  VOLUME = {70},
}

@BOOK{BeckerWeispfenning_1993,
  AUTHOR = {Becker, Thomas and Weispfenning, Volker},
  PUBLISHER = {Springer-Verlag},
  DATE = {1993},
  DOI = {10/cfhwn9},
  ISBN = {978-0-387-97971-7},
  LANGID = {english},
  SERIES = {Graduate {{Texts}} in {{Mathematics}}},
  SHORTTITLE = {Gröbner {{Bases}}},
  TITLE = {Gröbner {{Bases}}: {{A Computational Approach}} to {{Commutative Algebra}}},
}

@ARTICLE{BenderFaugereTsigaridas_2019,
  AUTHOR = {Bender, Matías R. and Faugère, Jean-Charles and Tsigaridas, Elias},
  PUBLISHER = {ACM},
  DATE = {2019-07-08},
  DOI = {10/gnt5z9},
  ISBN = {978-1-4503-6084-5},
  JOURNALTITLE = {Proc. {{ISSAC}} 2019},
  PAGES = {42--49},
  SHORTTITLE = {Gröbner Basis over Semigroup Algebras},
  TITLE = {Gröbner Basis over Semigroup Algebras: Algorithms and Applications for Sparse Polynomial Systems},
}

@ARTICLE{BerthomieuEderSafeyElDin_2021,
  AUTHOR = {Berthomieu, Jérémy and Eder, Christian and Safey El Din, Mohab},
  PUBLISHER = {ACM},
  DATE = {2021-07-18},
  DOI = {10/gk8549},
  ISBN = {978-1-4503-8382-0},
  JOURNALTITLE = {Proc. {{ISSAC}} 2021},
  PAGES = {51--58},
  SHORTTITLE = {Msolve},
  TITLE = {Msolve: A Library for Solving Polynomial Systems},
}

@ARTICLE{BosmaCannonPlayoust_1997,
  AUTHOR = {Bosma, Wieb and Cannon, John and Playoust, Catherine},
  DATE = {1997},
  DOI = {10/ckdngx},
  ISSN = {0747-7171},
  JOURNALTITLE = {J. Symb. Comput.},
  NUMBER = {3-4},
  PAGES = {235--265},
  TITLE = {The {{Magma}} Algebra System {{I}}: The User Language},
  VOLUME = {24},
}

@THESIS{Buchberger_1965,
  AUTHOR = {Buchberger, Bruno},
  INSTITUTION = {Insbruck University},
  DATE = {1965},
  LANGID = {ngerman},
  TITLE = {Ein Algorithmus zum Auffinden der Basiselemente des Restklassenringes nach einem nulldimensionalen Polynomideal},
}

@ARTICLE{Buchberger_2006,
  AUTHOR = {Buchberger, Bruno},
  TRANSLATOR = {Abramson, Michael P.},
  DATE = {2006-03-01},
  DOI = {10/dz9kz6},
  ISSN = {0747-7171},
  JOURNALTITLE = {J. Symb. Comput.},
  LANGID = {english},
  NUMBER = {3},
  ORIGDATE = {1965},
  PAGES = {475--511},
  SERIES = {Logic, {{Mathematics}} and {{Computer Science}}: {{Interactions}} in Honor of {{Bruno Buchberger}} (60th Birthday)},
  SHORTTITLE = {Bruno {{Buchberger}}’s {{PhD}} Thesis 1965},
  TITLE = {An Algorithm for Finding the Basis Elements of the Residue Class Ring of a Zero Dimensional Polynomial Ideal},
  VOLUME = {41},
}

@ARTICLE{CarusoVacconVerron_2020,
  AUTHOR = {Caruso, Xavier and Vaccon, Tristan and Verron, Thibaut},
  PUBLISHER = {ACM},
  DATE = {2020-07-20},
  DOI = {10/gp99pq},
  ISBN = {978-1-4503-7100-1},
  JOURNALTITLE = {Proc. {{ISSAC}} 2020},
  LANGID = {english},
  PAGES = {70--77},
  TITLE = {Signature-Based Algorithms for {{Gröbner}} Bases over {{Tate}} Algebras},
}

@BOOK{CoxLittleOShea_2015,
  AUTHOR = {Cox, David A. and Little, John and O'Shea, Donal},
  PUBLISHER = {Springer},
  DATE = {2015},
  DOI = {10/hzv6},
  EDITION = {4},
  ISBN = {978-3-319-16720-6},
  LANGID = {english},
  SERIES = {Undergraduate {{Texts}} in {{Mathematics}}},
  TITLE = {Ideals, {{Varieties}}, and {{Algorithms}}},
}

@MISC{Singular_CAS,
  AUTHOR = {Decker, Wolfram and Greuel, Gert-Martin and Pfister, Gerhard and Schönemann, Hans},
  URL = {http://www.singular.uni-kl.de},
  DATE = {2022},
  TITLE = {Singular 4-3-0 — {{A}} Computer Algebra System for Polynomial Computations},
}

@ARTICLE{EderFaugere_2017,
  AUTHOR = {Eder, Christian and Faugère, Jean-Charles},
  DATE = {2017-05-01},
  DOI = {10/ggck7f},
  ISSN = {0747-7171},
  JOURNALTITLE = {J. Symb. Comput.},
  NUMBER = {3},
  PAGES = {719--784},
  TITLE = {A Survey on Signature-Based Algorithms for Computing {{Gröbner}} Bases},
  VOLUME = {80},
}

@ARTICLE{EderLairezMohrSafeyElDin_2023,
  AUTHOR = {Eder, Christian and Lairez, Pierre and Mohr, Rafael and Safey El Din, Mohab},
  DATE = {2023-11-01},
  DOI = {10/jxn2},
  ISSN = {0747-7171},
  JOURNALTITLE = {J. Symb. Comput.},
  LANGID = {english},
  PAGES = {1--21},
  TITLE = {A Signature-Based Algorithm for Computing the Nondegenerate Locus of a Polynomial System},
  VOLUME = {119},
}

@ARTICLE{EderPerry_2011,
  AUTHOR = {Eder, Christian and Perry, John},
  PUBLISHER = {ACM},
  DATE = {2011-06-08},
  DOI = {10/dmwqmp},
  ISBN = {978-1-4503-0675-1},
  JOURNALTITLE = {Proc. {{ISSAC}} 2011},
  PAGES = {99--106},
  TITLE = {Signature-Based Algorithms to Compute {{Gröbner}} Bases},
}

@ARTICLE{EderPfisterPopescu_2017,
  AUTHOR = {Eder, Christian and Pfister, Gerhard and Popescu, Adrian},
  PUBLISHER = {ACM},
  DATE = {2017-07-23},
  DOI = {10/gsbxs2},
  ISBN = {978-1-4503-5064-8},
  JOURNALTITLE = {Proc. {{ISSAC}} 2017},
  PAGES = {141--148},
  SERIES = {{{ISSAC}} '17},
  TITLE = {On Signature-Based {{Gröbner}} Bases over {{Euclidean}} Rings},
}

@ARTICLE{EderRoune_2013,
  AUTHOR = {Eder, Christian and Roune, Bjarke Hammersholt},
  PUBLISHER = {ACM},
  DATE = {2013-06-26},
  DOI = {10/ggkppx},
  ISBN = {978-1-4503-2059-7},
  JOURNALTITLE = {Proc. {{ISSAC}} 2013},
  PAGES = {331--338},
  TITLE = {Signature Rewriting in {{Gröbner}} Basis Computation},
}

@ARTICLE{Faugere_1999,
  AUTHOR = {Faugère, Jean-Charles},
  DATE = {1999},
  DOI = {10/bpq5dx},
  ISSN = {0022-4049},
  JOURNALTITLE = {J. Pure Appl. Algebra},
  NUMBER = {1-3},
  PAGES = {61--88},
  TITLE = {A new efficient algorithm for computing {{Gröbner}} bases $(F_4)$},
  VOLUME = {139},
}

@ARTICLE{Faugere_2001,
  AUTHOR = {Faugère, Jean-Charles},
  PUBLISHER = {World Scientific},
  DATE = {2001},
  DOI = {10/d9297x},
  EVENTTITLE = {Proceedings of the {{Fifth Asian Symposium}} ({{ASCM}} 2001)},
  JOURNALTITLE = {Comput. {{Math}}.},
  LANGID = {english},
  PAGES = {1--12},
  TITLE = {Finding All the Solutions of {{Cyclic}} 9 Using {{Gröbner}} Basis Techniques},
}

@ARTICLE{Faugere_2002,
  AUTHOR = {Faugère, Jean-Charles},
  PUBLISHER = {ACM},
  DATE = {2002-07-10},
  DOI = {10/bd4nnq},
  ISBN = {978-1-58113-484-1},
  JOURNALTITLE = {Proc. ISSAC 2002},
  PAGES = {75--83},
  TITLE = {A new efficient algorithm for computing Gröbner bases without reduction to zero ($F_5$)},
}

@ARTICLE{FaugereJoux_2003,
  AUTHOR = {Faugère, Jean-Charles and Joux, Antoine},
  PUBLISHER = {Springer},
  DATE = {2003},
  DOI = {10/fpfzgc},
  ISBN = {978-3-540-45146-4},
  JOURNALTITLE = {{{CRYPTO}} 2003},
  LANGID = {english},
  PAGES = {44--60},
  SERIES = {{{LNCS}}},
  TITLE = {Algebraic Cryptanalysis of Hidden Field Equation ({{HFE}}) Cryptosystems Using {{Gröbner}} Bases},
}

@ARTICLE{FrancisVerron_2020,
  AUTHOR = {Francis, Maria and Verron, Thibaut},
  DATE = {2020-06-01},
  DOI = {10/gsbxs4},
  ISSN = {1661-8289},
  JOURNALTITLE = {Math.Comput.Sci.},
  LANGID = {english},
  NUMBER = {2},
  PAGES = {515--530},
  TITLE = {A Signature-Based Algorithm for Computing {{Gröbner}} Bases over Principal Ideal Domains},
  VOLUME = {14},
}

@ARTICLE{Galkin_2014,
  AUTHOR = {Galkin, V. V.},
  DATE = {2014-03-01},
  DOI = {10/ghjx58},
  ISSN = {1608-3261},
  JOURNALTITLE = {Program. Comput. Softw.},
  LANGID = {english},
  NUMBER = {2},
  PAGES = {47--57},
  TITLE = {Termination of the {{F5}} Algorithm},
  VOLUME = {40},
}

@ARTICLE{GaoGuanVolny_2010,
  AUTHOR = {Gao, Shuhong and Guan, Yinhua and Volny, Frank},
  PUBLISHER = {ACM},
  DATE = {2010-07-25},
  DOI = {10/cwg6rj},
  ISBN = {978-1-4503-0150-3},
  JOURNALTITLE = {Proc. {{ISSAC}} 2010},
  PAGES = {13--19},
  TITLE = {A New Incremental Algorithm for Computing {{Groebner}} Bases},
}

@ARTICLE{GaoVolnyWang_2016,
  AUTHOR = {Gao, Shuhong and Volny, Frank and Wang, Mingsheng},
  DATE = {2016-01},
  DOI = {10/f7t889},
  ISSN = {0025-5718, 1088-6842},
  JOURNALTITLE = {Math. Comp.},
  LANGID = {english},
  NUMBER = {297},
  PAGES = {449--465},
  TITLE = {A New Framework for Computing {{Gröbner}} Bases},
  VOLUME = {85},
}

@ARTICLE{GebauerMoller_1986,
  AUTHOR = {Gebauer, Rüdiger and Möller, H. Michael},
  PUBLISHER = {ACM},
  DATE = {1986-10-01},
  DOI = {10/cb24fn},
  ISBN = {978-0-89791-199-3},
  JOURNALTITLE = {Symp. {{Symb}}. {{Algebr}}. {{Comput}}.},
  PAGES = {218--221},
  SERIES = {{{SYMSAC}} '86},
  TITLE = {Buchberger's Algorithm and Staggered Linear Bases},
}

@ARTICLE{GebauerMoller_1988,
  AUTHOR = {Gebauer, Rüdiger and Möller, H. Michael},
  DATE = {1988-10-12},
  DOI = {10/bfjdwc},
  ISSN = {0747-7171},
  JOURNALTITLE = {J. Symb. Comput.},
  LANGID = {english},
  NUMBER = {2},
  PAGES = {275--286},
  TITLE = {On an Installation of {{Buchberger}}'s Algorithm},
  VOLUME = {6},
}

@INBOOK{GreenMoraUfnarovski_1998,
  AUTHOR = {Green, E. D. and Mora, Teo and Ufnarovski, Victor},
  EDITOR = {Bronstein, Manuel and Weispfenning, Volker and Grabmeier, Johannes},
  PUBLISHER = {Birkhäuser},
  BOOKTITLE = {Symb. {{Rewriting Tech}}.},
  DATE = {1998},
  DOI = {10/dfbt7g},
  ISBN = {978-3-0348-8800-4},
  LANGID = {english},
  PAGES = {93--104},
  SERIES = {Progress in {{Computer Science}} and {{Applied Logic}}},
  TITLE = {The Non-Commutative {{Gröbner}} Freaks},
}

@ARTICLE{HashemiArs_2010,
  AUTHOR = {Hashemi, Amir and Ars, Gwénolé},
  DATE = {2010-12-01},
  DOI = {10/bmfh29},
  ISSN = {0747-7171},
  JOURNALTITLE = {J. Symb. Comput.},
  LANGID = {english},
  NUMBER = {12},
  PAGES = {1330--1340},
  TITLE = {Extended {{F5}} Criteria},
  VOLUME = {45},
}

@ARTICLE{HashemiJavanbakht_2021,
  AUTHOR = {Hashemi, Amir and Javanbakht, Masoumeh},
  PUBLISHER = {World Scientific Publishing Co.},
  DATE = {2021-08},
  DOI = {10/gqwrpn},
  ISSN = {0219-4988},
  JOURNALTITLE = {J. Algebra Appl.},
  NUMBER = {8},
  PAGES = {2150132},
  TITLE = {On the Construction of Staggered Linear Bases},
  VOLUME = {20},
}

@ARTICLE{Higman_1952,
  AUTHOR = {Higman, Graham},
  DATE = {1952},
  DOI = {10/fmt8nh},
  ISSN = {1460-244X},
  JOURNALTITLE = {Proc. Lond. Math. Soc.},
  LANGID = {english},
  NUMBER = {1},
  PAGES = {326--336},
  SERIES = {3},
  TITLE = {Ordering by Divisibility in Abstract Algebras},
  VOLUME = {2},
}

@ARTICLE{HofstadlerVerron_2022,
  AUTHOR = {Hofstadler, Clemens and Verron, Thibaut},
  DATE = {2022-11-01},
  DOI = {10/gp85ss},
  ISSN = {0747-7171},
  JOURNALTITLE = {J. Symb. Comput.},
  LANGID = {english},
  PAGES = {211--241},
  TITLE = {Signature {{Gröbner}} Bases, Bases of Syzygies and Cofactor Reconstruction in the Free Algebra},
  VOLUME = {113},
}

@ARTICLE{HofstadlerVerron_2023,
  AUTHOR = {Hofstadler, Clemens and Verron, Thibaut},
  PUBLISHER = {Association for Computing Machinery},
  DATE = {2023-07-24},
  DOI = {10/gskrd2},
  ISBN = {9798400700392},
  JOURNALTITLE = {Proc. {{ISSAC}} 2023},
  PAGES = {298--306},
  SERIES = {{{ISSAC}} '23},
  TITLE = {Signature {{Gröbner}} Bases in Free Algebras over Rings},
}

@ARTICLE{Huet_1980,
  AUTHOR = {Huet, Gérard},
  DATE = {1980},
  DOI = {10/fj7n4g},
  ISSN = {0004-5411},
  JOURNALTITLE = {J. ACM},
  NUMBER = {4},
  PAGES = {797--821},
  TITLE = {Confluent Reductions: Abstract Properties and Applications to Term Rewriting Systems},
  VOLUME = {27},
}

@ONLINE{Kambe_2023,
  AUTHOR = {Kambe, Yuta},
  URL = {http://arxiv.org/abs/2305.13639},
  DATE = {2023-07-19},
  EPRINT = {2305.13639},
  EPRINTCLASS = {math},
  EPRINTTYPE = {arxiv},
  LANGID = {english},
  TITLE = {Analysis of Computing {{Gr}}\textbackslash "obner Bases and {{Gr}}\textbackslash "obner Degenerations via Theory of Signatures},
  URLDATE = {2023-11-28},
}

@BOOK{KreuzerRobbiano_2000,
  AUTHOR = {Kreuzer, Martin and Robbiano, Lorenzo},
  PUBLISHER = {Springer},
  DATE = {2000},
  DOI = {10/ffxbqr},
  ISBN = {978-3-540-67733-8},
  LANGID = {english},
  TITLE = {Computational Commutative Algebra},
  VOLUME = {1},
}

@ARTICLE{LuWangXiaoZhou_2018,
  AUTHOR = {Lu, Dong and Wang, Dingkang and Xiao, Fanghui and Zhou, Jie},
  PUBLISHER = {ACM},
  DATE = {2018-07-11},
  DOI = {10/gsbxr9},
  ISBN = {978-1-4503-5550-6},
  JOURNALTITLE = {Proc. {{ISSAC}} 2018},
  PAGES = {271--278},
  SERIES = {{{ISSAC}} '18},
  TITLE = {Extending the {{GVW}} Algorithm to Local Ring},
}

@ARTICLE{MollerMoraTraverso_1992,
  AUTHOR = {Möller, H. Michael and Mora, Teo and Traverso, Carlo},
  PUBLISHER = {ACM},
  DATE = {1992},
  DOI = {10/cgb2ts},
  ISBN = {978-0-89791-489-5},
  JOURNALTITLE = {Proc. {{ISSAC}} 1992},
  LANGID = {english},
  PAGES = {320--328},
  TITLE = {Gröbner Bases Computation Using Syzygies},
}

@ARTICLE{MonaganPearce_2015,
  AUTHOR = {Monagan, Michael and Pearce, Roman},
  PUBLISHER = {ACM},
  DATE = {2015-07-10},
  DOI = {10/ggpbmk},
  ISBN = {978-1-4503-3599-7},
  JOURNALTITLE = {Proc. {{PASCO}} 2015},
  PAGES = {95--100},
  TITLE = {A Compact Parallel Implementation of {{F4}}},
}

@ARTICLE{Mora_1994,
  AUTHOR = {Mora, Teo},
  DATE = {1994-11-07},
  DOI = {10/dvwxsv},
  ISSN = {0304-3975},
  JOURNALTITLE = {Theor. Comput. Sci.},
  LANGID = {english},
  NUMBER = {1},
  PAGES = {131--173},
  TITLE = {An Introduction to Commutative and Noncommutative {{Gröbner}} Bases},
  VOLUME = {134},
}

@BOOK{Mora_2005,
  AUTHOR = {Mora, Teo},
  PUBLISHER = {Cambridge University Press},
  DATE = {2005},
  DOI = {10/jdwm},
  SERIES = {Encyclopedia of {{Mathematics}} and Its {{Applications}}},
  TITLE = {Solving Polynomial Equation Systems},
  VOLUME = {2},
}

@ARTICLE{RouneStillman_2012,
  AUTHOR = {Roune, Bjarke Hammersholt and Stillman, Michael},
  PUBLISHER = {ACM},
  DATE = {2012-07-22},
  DOI = {10/ggkpqd},
  ISBN = {978-1-4503-1269-1},
  JOURNALTITLE = {Proc. {{ISSAC}} 2012},
  PAGES = {203--210},
  TITLE = {Practical {{Gröbner}} Basis Computation},
}

@ARTICLE{Stillman_1990,
  AUTHOR = {Stillman, Mike},
  DATE = {1990-10-01},
  DOI = {10/dsx7mz},
  ISSN = {1572-9036},
  JOURNALTITLE = {Acta Appl. Math.},
  LANGID = {english},
  NUMBER = {1},
  PAGES = {77--103},
  TITLE = {Methods for Computing in Algebraic Geometry and Commutative Algebra},
  VOLUME = {21},
}

@ONLINE{SunWang_2011,
  AUTHOR = {Sun, Yao and Wang, Dingkang},
  DATE = {2011-08-05},
  EPRINT = {1108.1301},
  EPRINTTYPE = {arxiv},
  TITLE = {Solving Detachability Problem for the Polynomial Ring by Signature-Based {{Gröbner}} Basis Algorithms},
}

@ARTICLE{SunWang_2013,
  AUTHOR = {Sun, Yao and Wang, Dingkang},
  DATE = {2013-04-01},
  DOI = {10/gp867m},
  ISSN = {1869-1862},
  JOURNALTITLE = {Sci. China Math.},
  LANGID = {english},
  NUMBER = {4},
  PAGES = {745--756},
  TITLE = {A New Proof for the Correctness of the {{F5}} Algorithm},
  VOLUME = {56},
}

@ARTICLE{SunWangMaZhang_2012,
  AUTHOR = {Sun, Yao and Wang, Dingkang and Ma, Xiaodong and Zhang, Yang},
  PUBLISHER = {ACM},
  DATE = {2012-07-22},
  DOI = {10/ghtkmx},
  JOURNALTITLE = {Proc. {{ISSAC}} 2012},
  PAGES = {351--358},
  TITLE = {A Signature-Based Algorithm for Computing {{Gröbner}} Bases in Solvable Polynomial Algebras},
}

@ARTICLE{Traverso_1996,
  AUTHOR = {Traverso, Carlo},
  DATE = {1996-10-01},
  DOI = {10/b3x2ct},
  ISSN = {0747-7171},
  JOURNALTITLE = {J. Symb. Comput.},
  LANGID = {english},
  NUMBER = {4},
  PAGES = {355--376},
  TITLE = {Hilbert Functions and the {{Buchberger}} Algorithm},
  VOLUME = {22},
}

@ARTICLE{VacconVerronYokoyama_2018,
  AUTHOR = {Vaccon, Tristan and Verron, Thibaut and Yokoyama, Kazuhiro},
  PUBLISHER = {ACM},
  DATE = {2018-07-11},
  DOI = {10/d3mr},
  ISBN = {978-1-4503-5550-6},
  JOURNALTITLE = {Proc. {{ISSAC}} 2018},
  PAGES = {383--390},
  SERIES = {{{ISSAC}} '18},
  TITLE = {On Affine Tropical {{F5}} Algorithms},
}

@ARTICLE{VacconYokoyama_2017,
  AUTHOR = {Vaccon, Tristan and Yokoyama, Kazuhiro},
  PUBLISHER = {ACM},
  DATE = {2017-07-23},
  DOI = {10/gsbxs8},
  ISBN = {978-1-4503-5064-8},
  JOURNALTITLE = {Proc. {{ISSAC}} 2017},
  PAGES = {429--436},
  SERIES = {{{ISSAC}} '17},
  TITLE = {A Tropical {{F5}} Algorithm},
}

@BOOK{Winkler_1996,
  AUTHOR = {Winkler, Franz},
  PUBLISHER = {Springer-Verlag},
  DATE = {1996},
  DOI = {10/bkh6hq},
  ISBN = {978-3-211-82759-8},
  LANGID = {english},
  SERIES = {Texts \& {{Monographs}} in {{Symbolic Computation}}},
  TITLE = {Polynomial Algorithms in Computer Algebra},
}

\end{document}